\documentclass[11 pt]{article}
\usepackage{amsmath, mathtools}
\usepackage{amsthm}
\usepackage{amssymb}
\usepackage{graphicx}
\usepackage{appendix}
\usepackage{color}
\usepackage{float}
\usepackage{hyperref}
\usepackage{bbm}

\newtheorem{definition}{Definition}[section]
\newtheorem{assumption}[definition]{Assumption}
\newtheorem{theorem}[definition]{Theorem}
\newtheorem{lemma}[definition]{Lemma}
\numberwithin{equation}{section}
\newtheorem{remark}[definition]{Remark}
\begin{document}

\title{Some asymptotics for short maturity Asian options}
\author{Humayra Shoshi$^1$ \and Indranil SenGupta$^2$$^,$$^3$\footnote{Corresponding author. Email: isengupt@fiu.edu}}
\date{%
    $^1$ Assistant Vice President (Model/Anlys/ Valid Sr Analyst), \\ Citigroup Inc., Tampa, Florida.\\%
    $^2$ Professor of Environmental Finance \& Risk Management in the Institute of Environment, \\ Florida International University, Miami, Florida. \\
    $^3$ Professor of Department of Mathematics and Statistics, \\ Florida International University, Miami, Florida. \\[2ex]%
   \today
}


\maketitle

\begin{abstract}
Most of the existing methods for pricing Asian options are less efficient in the limit of small maturities and small volatilities. In this paper, we use the large deviations theory for the analysis of short-maturity Asian options. We present a constant volatility model for the underlying market that incorporates a jump term in addition to the drift and diffusion terms. We estimate the asymptotics for the out-of-the-money, in-the-money, and at-the-money short-maturity Asian call and put options. Under appropriate assumptions, we show that the asymptotics for out-of-the-money Asian call and put options are governed by rare events. For the at-the-money Asian options, the result is more involved and in that case, we find the upper and lower bounds of the asymptotics of the Asian option price.

\end{abstract}



	
\section{Introduction}

An Asian option is a type of \text{exotic option} where the payoff depends on the average price of the underlying asset. This is different than the standard American or European options, where the payoff depends on the price of the underlying asset at a specific point in time, also known as the time to maturity. Typically, an option contract specifies the \text{average} price, which could be a geometric or arithmetic average of the price of the underlying asset at discrete time intervals. Traders are interested in such an option due to its relatively low volatility which is possible due to the averaging mechanism. Traders who are exposed to the underlying asset over a certain period of time tend to utilize Asian options. An example would be the commodity market since these options are less expensive than American or European options as the volatility of the average price is less than that of the spot price.

In the literature, a broad range of work can be found on the asymptotics of the option prices for short-maturity European options. In \cite{bbf, chengetal} the authors obtained closed-form asymptotic formulae for the local volatility models. In \cite{figlopez} the authors used exponential L\'evy models to study the asymptotics of the implied volatility. In \cite{alos, mforde} the authors study the short-time behavior of the implied volatility for the stochastic volatility models. In contrast, the study of short-maturity Asian options has been gaining interest very recently. In \cite{pirjol}, the authors discuss the short-maturity asymptotics for the price of the Asian options under the assumption that the underlying asset price follows a local volatility model.

An abundance of work done on the pricing of Asian options in the general setting (i.e., not necessarily for the short maturity) can be found in the literature. The pricing under the Black-Scholes model has been studied in \cite{carr, geman,  linetsky}. The authors in \cite{geman} used Bessel processes to solve the integral of an exponential Brownian motion, and implemented Laplace transform results for the pricing of Asian options. Another popular method, known for a very long time and also carries the advantage that it has a broad range of applications to other models, is the partial differential equation (PDE) approach (see \cite{rogersshi, vecer1, vecer2}). In \cite{vecer1}, the author characterized the price of the Asian option by using a one-dimensional PDE that could be used for both continuous and discrete average Asian options. The resulting PDE can be solved numerically (see \cite{vecer1, vecer2}). In \cite{sudip}, the authors develop a pricing partial integro-differential equation and its Fourier transform expression for floating Asian options. The primary contribution of that paper is to develop a closed-form pricing expression for the contract. It is also shown that the procedure is easy to implement for all classes of L\'evy processes. In \cite{Mellin} the authors derive an expression for the floating strike put arithmetic Asian options based on the technique of the Mellin transform.

It is to be noted that the existing methods found in the literature regarding pricing Asian options are numerically less efficient in the small maturity and small volatility regime. In \cite{carr}, the authors discuss the drawback of the Geman-Yor method (found in \cite{geman}), for the case of the Asian options under the Black-Scholes model that was formulated using Laplace transform. The use of the inverse Laplace transform in pricing small-maturity Asian options is more complex as discussed in \cite{dufrense, fuetal}. A similar problem can be observed in the use of the spectral method (see \cite{linetsky}), where the authors' formulation of an integral involved incomplete gamma functions and Laguerre polynomials. However, in \cite{foschip} the above issues are overcome by the authors when pricing Asian options under a local volatility model. They used the natural geometric-differential structure of the pricing operator regarded as a hypoelliptic (not uniformly parabolic) PDE of Kolmogorov type in $\mathbb{R}^3$, and successfully developed approximate formulae in terms of elementary functions for the price of Asian options in a general local volatility model, which complements some of the numerically more complex and less efficient alternative approaches.

In this work, we will use the large deviations theory for small-time diffusion processes. Our goal is to obtain analytical results for the short-maturity asymptotics of Asian options in the volatility model with an associated L\'evy subordinator. The paper \cite{pirjol} is a primary motivator for this work. For the present paper, the underlying volatility model is driven by  both a Brownian motion and a L\'evy subordinator. This underlying model is closely connected to a generalized Barndorff-Nielsen and Shephard (BN–S) model (see \cite{ijtaf}). It is worth noting that similar models are successfully implemented in recent literature for derivative and commodity market analysis (see \cite{hum, hum2}).  The advantage of the large deviation principle is that once certain conditions are satisfied one can then apply the contraction principle to get the corresponding large deviations for the small time arithmetic average of the diffusion process and therefore formulate a robust expression to the asymptotic behavior for the out-of-the-money Asian call and put options. The asymptotic exponent obtained is the rate function from the large deviation principle, which itself is a complex and non-trivial variational problem. We will notice that the asymptotics for the in-the-money case can be found by using the put-call parity. Also, the asymptotics for at-the-money short-maturity Asian options will be computed.

The rest of the paper proceeds as follows. In Section \ref{sec2}, we provide definitions and mathematical preliminaries- especially important results from large deviation theory-  that are used in the  rest of the paper. In Section \ref{sec3} we describe the underlying model and assumptions with respect to a risk-neutral probability measure. In Section \ref{sec4}, we derive various main results. We start that section with a pricing formula for short maturity Asian options. After that we study in detail short maturity out-of-the-money and short maturity in-the money Asian options, followed by short maturity at-the-money Asian options. In that section various analytical results related to asymptotics for out-of-the-money, in-the-money, and at-the-money cases are derived for fixed strike Asian options. Finally, some numerical results and a brief conclusion are provided in Section \ref{sec5}.

\section{Mathematical preliminaries}
\label{sec2}
The large deviations theory concerns itself with the asymptotic 
behavior, or roughly speaking the exponential decline, of the probability 
measures of certain kinds of rare events. The theory can be effectively applied 
to gather information from a probabilistic model, with applications found in 
information theory, risk management and statistical mechanics to name a few. 

In our analysis we are interest in asking large deviations questions that involve sample path 
of stochastic processes. Specifically, if $X^{\epsilon}_t$ denotes a family of 
processes that converge, as $\epsilon \to 0$, to some deterministic limit, we 
may be interested to learn what the rate of 
convergence is.  Let $(\Omega, \mathcal{F}, \mathbb{P})$ denote a completed probability space, 
where $\Omega$ is a topological space such that open and closed subsets of 
$\Omega$ are well-defined, $\mathcal{F}$ is a $\sigma$- field that need not be 
a Borel $\sigma$-field, $\mathbb{P}$ is a probability measure on  
$\mathcal{F}$. Here, $\mathcal{F}_{\Omega}$ denotes the completed Borel 
$\sigma$-field on $\Omega$.

\begin{definition}
\label{ch2sec25def1}
\emph{(Lower semicontinuous mapping)} A function $F: \Omega \to [-\infty, 
\infty]$ is called \emph{lower semicontinuous} at a point $x' \in \Omega$ if for every 
real $a < F(x')$ there exists a neighborhood $U$ of $x'$ such that $F(x) > a$ 
for all $x \in U$. Equivalently, $F$ is lower semicontinuous at $x' \iff 
\liminf_{x \to x'}F(x) \geq F(x').$
\end{definition}

Paraphrasing the definition of the rate function from \cite{dembo}:
\begin{definition}
\label{ch2sec25def2}
A \emph{rate function} I is a lower semicontinuous mapping $I: \Omega \to [0, 
\infty]$, such that for all $\beta \in [0, \infty )$, the level set 
$\Psi_{I}(\beta):= \{x: I(x) \leq \beta \}$ is a closed subset of $\Omega$. A 
good rate function is a rate function for which all the level sets 
$\Psi_{I}(\beta)$ are compact subsets of $\Omega$. 
\end{definition}

\begin{definition}
The \emph{effective domain} of $I$, denoted $\mathcal{D}_{I},$ is the set of points in 
$\Omega$ of finite rate, namely, $\mathcal{D}_{I} := \{x: I(x) < \infty \}$. 
When no confusion occurs, we refer to $\mathcal{D}_{I}$ as the domain of I.
\end{definition}


\begin{definition}
	(Large deviation principle) A sequence $(\mathbb{P}_{\epsilon})_{\epsilon 
	\in \mathbb{R}^{+}}$ of probability measures on a topological space 
	$\mathcal{U}$ satisfies the large deviation principle (LDP) with the rate 
	function $I : \mathcal{U} \to \mathbb{R}$ if \textit{I} is nonnegative and 
	lower semicontinuous and for any measurable set \textit{A}, we have
	\begin{equation}
		\label{def21eqn1}
		-\inf_{x \in A^{\circ}} I(x) \leq \liminf_{\epsilon \to 0} \epsilon\log 
		\mathbb{P}_{\epsilon}(A) \leq \limsup_{\epsilon \to 0} \epsilon\log 
		\mathbb{P}_{\epsilon}(A) \leq - \inf_{x \in \bar{A}} I(x)
	\end{equation}
	Here, $A^{\circ}$ and $\bar{A}$ denote the interior and closure of $A$, 
	respectively.
\end{definition}

The right- and left-hand sides of \eqref{def21eqn1} are referred to as the upper 
and lower bounds, respectively.
\begin{remark} If $\{\mathbb{P}_{\epsilon}\}$ satisfies the LDP and $A \in 
\mathcal{F}$ is such that $\inf_{x \in A^{\circ}} I(x) = \inf_{x \in \bar{A}} 
I(x) := I_{A}$, then $\lim_{\epsilon \to 0} \epsilon \log \mathbb{P}_{\epsilon}(A) = -I_{A}$, where the set $A$ satisfying the condition is called an \textit{I continuity set}.
\end{remark}

There are several alternate formulations of \eqref{def21eqn1} in the literature, that 
are found to be more useful when proving it, and we are interested in one such 
alternative expressions for the upper and lower bounds, when $\mathcal{F}_{\Omega} 
\subseteq \mathcal{F}$. 

When $\mathcal{F}_{\Omega} \subseteq \mathcal{F}$, the LDP is equivalent to the 
following bounds:
\begin{enumerate}
	\item[(i)] (Upper bound) For any closed set $A \subseteq \Omega$, 
	\begin{equation*}
		\limsup_{\epsilon \to 0} \epsilon \log\mathbb{P}_{\epsilon}(A) \leq 
		-\inf_{x \in A} I(x).
	\end{equation*}
    \item[(ii)] (Lower bound) For any open set $B \subseteq \Omega$, 
    \begin{equation*}
    	\label{ch2sec242}
    	\liminf_{\epsilon \to 0} \epsilon \log\mathbb{P}_{\epsilon}(B) \geq 
    	-\inf_{x \in B} I(x). 
    \end{equation*}
\end{enumerate}

We end this section with an important theorem and a corollary that we will later use in this paper to compute the rate function of a short 
maturity out-of-the-money Asian call and put options. The following theorem known as the \textit{contraction principle} is a powerful 
tool which shows that under continuous mappings the LDP is preserved, although, possibly, the rate function could change. For details please see \cite{dembo}. Later in the paper we will see 
the application of this theorem in pricing the out-of-the-money Asian call and 
put options. 

\begin{theorem}
	\label{thm22}
	(Contraction Principle). If $\mathbb{P}_{\epsilon}$ satisfies a large 
	deviation principle on $\mathcal{U}$ with rate function $I(x)$ and $F : 
	\mathcal{U} \to \mathcal{V}$ is a continuous map, then the probability 
	measure $\mathbb{Q}_{\epsilon} :=\mathbb{P}_{\epsilon} F^{-1} $ satisfies a 
	large deviation principle on $\mathcal{V}$ with rate function
	\begin{equation}
		\label{thm22eqn1}
		J(y) = \inf_{x: F(x) =y} I(x)
	\end{equation}
\end{theorem}

Before the next theorem, we enlarge the underlying probability 
space $(\Omega, \mathcal{F}, \mathbb{P})$. Define $\bar{\Omega} := \mathbb{R} 
\times \Omega,  \, \bar{\mathcal{F}}:= \mathcal{B}(\mathbb{R}) \otimes 
\mathcal{F}, \, \mathbb{P}^x:= \delta_{x} \otimes \mathbb{P}$.
From \cite{kuhn} (Example 6.6 and Theorem 7.1) we have the following theorem.
\begin{theorem}
\label{ch2sec24cor1}
Let $b, \sigma, \eta : \mathbb{R} \to \mathbb{R}$ be bounded, locally Lipschitz 
continuous functions, i.e. for any $M > 0$ there exists $L = L(M)>0, L' = L'(M) 
>0, L''= L''(M)>0$ such that $$|b(x)-b(y)| \leq L|x-y|, \quad |\sigma(x) - 
\sigma(y)|\leq L'|x-y|, \quad |\eta(x)-\eta(y)| \leq L''|x-y|, \, \emph{for 
all }  \, x, y \in [0, M].$$

 Let $(Z)_{t \geq 0}$ be a L\'evy process with 
L\'evy triplet $(0, 0, \nu)$ such that $\mathbb{E}^{\mathbb{P}^{x}} [e^{\lambda 
|Z_{1}|}] < \infty$ for all $\lambda \geq 0$, i.e. $Z_{1}$ satisfies the 
exponential integrability condition for large deviations with respect to the 
uniform topology. If $$g(x, \xi) = b(x) \xi + \frac{1}{2} \sigma^2(x)\xi^2 + 
\int_{\mathbb{R}\setminus\{ 0\}} (e^{y \eta(x) \xi} -1 -y\eta(x) \xi 
\mathbbm{1}_{|y|\leq 1})d\nu(y),$$ satisfies the given conditions,
\begin{enumerate}
\item[(i)] $g^{\ast}(x, \alpha) < \infty$ for all $x, \alpha \in \mathbb{R}$; for any $N >0$ there exists constants $c_1, c_2 >0$ such that $$g^{\ast}(x, \alpha) + \left|\frac{\partial}{\partial \alpha} g^{\ast}(x, \alpha) \right| \leq c_1, \quad \emph{and} \quad \frac{\partial^2}{\partial \alpha^2} g^{\ast}(x, \alpha) > c_2 \quad \emph{for all} \quad x \in \mathbb{R}, |\alpha| \leq N.$$

\item[(ii)]Continuity condition:
$$\Delta g^{\ast}(\delta):= \sup_{|x-y|< \delta} \sup_{\alpha \in \mathbb{R}} \frac{g^{\ast}(x, \alpha)- g^{\ast}(y, \alpha)}{1 + g^{\ast}(y, \alpha)} \to 0, \, \emph{as}\,  \delta \to 0.$$
\end{enumerate}
Then the family $(X^{\epsilon})_{\epsilon>0}$ of solutions on $ (\bar{\Omega}, \bar{\mathcal{F}}, \mathbb{P}^{x})$, 
\begin{equation}
\label{ch2sec24cor2}
dX^{\epsilon}_{t} = b(X^{\epsilon}_{{t}_{-}})dt + \sqrt{\epsilon}\sigma(X^{\epsilon}_{{t}_{-}})dW_t + \eta(X^{\epsilon}_{{t}_{-}})dZ^{\epsilon}_{t}, \quad Z^{\epsilon}_{t}:= \epsilon Z_{\frac{t}{\epsilon}}. 
\end{equation}
obeys a large deviation principle in $(D[0,1], \|\cdot\|_{\infty})$ with good rate function
\[\mathcal{I}(x, g) 
:= 
\begin{cases}
	\int_{0}^{1} g^{\ast}(\phi(t), \phi'(t))dt, & \phi \in AC[0,1], \phi(0) = x\\
	\infty, & otherwise.\\
\end{cases}
\]
where $g^{\ast}(x, \cdot)$ denotes the Legendre transform of the convex function $g(x, \cdot)$, and 
$AC[0,1]$ is the space of absolutely continuous functions on $[0, 1]$.
\end{theorem}
The above result is further refined in Theorem 7.6 of \cite{kuhn}. Comparing these two results, it can be shown that the rate functions are the same in both cases.

Suppose $g(0)= \log S_0$, and $g \in AC[0,1]$. Note that the map $g \to \int_0^1 e^{g(x)}\,dx$, is a continuous map $L_{\infty}[0,1] \to \mathbb{R}+$. By the Theorem \ref{thm22}, $\mathbb{P}\left(\int_0^1 e^{X_{tT}^{\epsilon}} \in \cdot \right)$ satisfies a large deviation principle with the rate function, that can be denoted as  
\begin{equation}
\label{koya}
\mathcal{I}(x, S_0)= \inf_{\int_0^1 e^{g(t)}\,dt=x, g(0)= \log S_0, g \in AC[0,1]} \mathcal{I}(x, g).
\end{equation}
The notation used on the left hand side of \eqref{koya} is used in \cite{pirjol} in the context of ``non-jump" models.  For the rest of this paper we will use the same notation. We conclude this section with a couple of other results that are used in this paper. 
\begin{theorem}
\label{ch2thmburk}
\emph{(Burkholder-Davis-Gundy inequality)}
Let $(\Omega, \mathcal{F}, (\mathcal{F}_{t})_{t \geq 0}, \mathbb{P})$ be a filtered probability space. For any $1 \leq p < \infty$ there exists positive constants $c_p, \, C_p$ such that for all local martingales $X$ with $X_0 = 0$ and stopping times $\tau$, the following inequality holds: $$c_p \mathbb{E} \left[ [X]^{\frac{p}{2}}_{\tau}\right] \leq \mathbb{E} \left[ [X^{\ast}]^{p}_{t}\right] \leq C_p \mathbb{E} \left[ [X]^{\frac{p}{2}}_{\tau}\right], $$ where $[X]_t$ is the quadratic variation of $X_t$. Furthermore, for continuous local martingales this statement holds for all $p \in (0, \infty).$
\end{theorem}

\begin{theorem}
\label{ch2thmdoob}
\emph{(Doob's	maximal inequalities)}  Let $(\mathcal{F}_t)_{t \geq 0}$ be a filtration on the probability space $(\Omega, \mathcal{F}, (\mathcal{F}_t)_{t \geq 0}, \mathbb{Q})$ and let $(M_t)_{t\geq 0}$ be a continuous martingale with respect to the filtration $(\mathcal{F}_t)_{t \geq 0}$.
\begin{itemize}
\item[(i)] Let $p \geq 1$ and $T>0$. If $\mathbb{E}^{\mathbb{Q}}\left[ |M_T|^{p} \right] < +\infty$ 	then $$\mathbb{Q}\left( \sup_{0\leq t \leq T} |M_T| \geq \lambda \right) \leq\frac{\mathbb{E}^{\mathbb{Q}}\left[ |M_T|^{p} \right]}{\lambda^p}.$$
\item[(ii)] Let $p > 1$ and $T>0$. If $\mathbb{E}^{\mathbb{Q}}\left[ |M_T|^{p} \right] < +\infty$ 	then $$\mathbb{E}^{\mathbb{Q}}\left[ \left(\sup_{0\leq t \leq T} |M_T|\right)^p \right] \leq		\left(\frac{p}{p-1}\right)^2 \mathbb{E}^{\mathbb{Q}}\left[|M_T|^{p} \right].$$
\end{itemize}
\end{theorem}

\section{Model and assumptions}
\label{sec3}
Let $(\Omega, \mathcal{F}, (\mathcal{F}_t)_{0 \leq t \leq T}, \mathbb{P})$ be 
some filtered probability space, where $\mathbb{P}$ is a risk-neutral 
probability measure. We assume that the log-stock price process of the stock $S = 
(S_t)_{t\geq 0}$ on $(\Omega, \mathcal{F}, (\mathcal{F}_t)_{0 \leq t \leq T}, 
\mathbb{P})$ is given by:
\begin{equation}
\label{ch5sec51}
X_t = \log S_t,
\end{equation}
where
\begin{equation}
\label{ch5sec52}
dX_t = (r - q) dt + \sigma(e^{X_t})dW_t + \rho dZ_{\lambda t}.
\end{equation}
In the above expressions, $r, \, q \geq 0$ are the risk-free interest rate and 
continuous dividend yield respectively, where a stock is traded up to a fixed 
horizon date $T$. The parameters $\rho \leq 
0$, $\lambda > 0$, and $\sigma(\cdot)$ is the local volatility. $W_t$ is a 
standard Brownian motion and $Z$ is a subordinator. In addition, $W \, 
\mathrm{and} \, \, Z$ are assumed to be independent and $(\mathcal{F}_t)_{0 
\leq t \leq T}$ is assumed to be the usual augmentation of the filtration 
generated by the pair $(W, Z)$. 

In \cite{pirjol2}, a local volatility function  $\sigma(\cdot)$ is assumed which satisfies
\begin{gather}
\label{sec23}
0 < \underline{\sigma} \leq \sigma(\cdot) \leq \bar{\sigma} < \infty,  \\
\label{sec24new}
|\sigma(e^{x}) - \sigma(e^{y})| \leq M |x-y|^{\alpha}, \end{gather}
for any $x$ and $y$, where $M, \, \alpha > 0 $  and $0< \underline{\sigma} < \bar{\sigma} <\infty$ are some fixed constants.

In this paper, we assume:
\begin{assumption}
\label{as512}
Under the small maturity regime the local volatility $\sigma(S_t) \approx \sigma$, i.e., (approximately) a constant.
\end{assumption}
Let the random measure associated with the jumps of $Z$, and L\'evy density of $Z$ be given by $J_Z$ and $\nu_Z$ respectively. In \cite{NV} the dynamics of $S_t$ is obtained based on models similar to \eqref{ch5sec51} and \eqref{ch5sec52}.  In our case, in view of \eqref{ch5sec51} and \eqref{ch5sec52}, with respect to $\mathbb{P}$, the dynamics of $S_t$ is given by:
\begin{equation}
\label{ch5sec21}
\frac{dS_t}{S_t} = (r-q + \frac{1}{2}\sigma^{2})dt + \sigma dW_t + 
\int_{\mathbb{R}^{+}}(e^{\rho x}-1){J}_{Z}(\lambda dt, dx), \quad S_0 > 0,
\end{equation}
where the parameters are described above and $S_0 >0$ is called the spot price.

\begin{assumption}
\label{ch5as5224}
Consider $t \in [0,T]$. With respect to measure $\mathbb{P}$, assume that
\begin{equation}
\label{ch5eq521}
\underline{\sigma}' t \leq \int_0^t \int_{\mathbb{R}^{+}}\left(e^{\rho x}-1 \right){J}_{Z}(\lambda ds, dx)  \leq \bar{\sigma}' t, 
\end{equation}
and $\underline{\mu} dt \leq \left(\int_{\mathbb{R}^{+}}(e^{\rho x}-1){J}_{Z}(\lambda dt, dx)\right) ^2  \leq \bar{\mu} dt$, which can be written as:
\begin{equation}
\label{ch5eq522}
\underline{\mu} t \leq \int_0^t \int_{\mathbb{R}^{+}}\left(e^{\rho x}-1 \right)^2 {J}_{Z}(\lambda ds, dx)\leq \bar{\mu}t,
\end{equation}
hold true, where $-\infty  < \underline{\sigma}' < \bar{\sigma}' < \infty, \, \emph{and}\, \, 0< \underline{\mu}< \bar{\mu} <\infty, \, J_Z$ is the jump measure of the subordinator $Z$, $\lambda > 0, \, \rho \leq 0, \, x\in \mathbb{R}^{+}.$
\end{assumption}
It follows from \eqref{ch5eq521} and \eqref{ch5eq522} that:
\begin{equation}
\label{newGia0}
\underline{\sigma}'  \leq \lambda \int_{\mathbb{R}^{+}}\left(e^{\rho x}-1 \right) \nu_{Z}(dx)  \leq \bar{\sigma}', 
\end{equation}
and
\begin{equation}
\label{newGia}
\underline{\mu} \leq \lambda \int_{\mathbb{R}^{+}}\left(e^{\rho x}-1\right)^2 \nu_{Z} (dx)  \leq \bar{\mu}.
\end{equation}

We also note that under the risk-neutral measure $\mathbb{P}$, the process $Y_t = e^{-(r-q)t} S_t$ is a martingale. Consequently, we have the following condition that we list as an assumption.
\begin{assumption}
	\label{asmp528}
The choice of $\lambda>0$ is obtained from $\frac{\sigma^2}{2} + \lambda 
\int_{\mathbb{R}^{+}} \left(e^{\rho x}-1\right)\nu_{Z} (dx) = 
0$. 
\end{assumption}
\begin{remark}
\label{remm}
In view of \eqref{newGia0} and Assumption \ref{asmp528}, we have $\underline{\sigma}', \bar{\sigma}' <0$. Consequently, from \eqref{ch5eq521}
\begin{equation}
\label{koyell}
- \bar{\sigma}' t \leq \left|\int_0^t \int_{\mathbb{R}^{+}}\left(e^{\rho x}-1 \right){J}_{Z}(\lambda ds, dx)  \right| \leq  -\underline{\sigma}' t, 
\end{equation}
\end{remark}
\begin{remark}
With  Assumption \ref{ch5as5224} and Remark \ref{remm} we can obtain:
\begin{align*}
\left| \int_{0}^{t}\int_{\mathbb{R}^{+}} \left(e^{\rho x}-1 
\right)  \tilde{J}_Z(\lambda du, dx) \right| & = \left|\int_{0}^{t}\int_{\mathbb{R}^{+}} \left(e^{\rho x}-1 
\right) J_Z(\lambda du, dx)
- \int_{0}^{t}\int_{\mathbb{R}^{+}} \left(e^{\rho x}-1 
\right)  \nu_Z (dx) \lambda du \right| \\
& \leq  \left|\int_{0}^{t}\int_{\mathbb{R}^{+}} \left(e^{\rho x}-1 
\right) J_Z(\lambda du, dx)\right| + \left|\int_{0}^{t}\int_{\mathbb{R}^{+}} \left(e^{\rho x}-1 
\right)  \nu_Z (dx) \lambda du \right| \\
& \leq  -\int_{0}^{t}  \underline{\sigma}' du - \left|\int_{0}^{t} \lambda du 
\int_{\mathbb{R}^{+}} \left(e^{\rho x}-1 \right) \nu_Z 
(dx)  \right| \\
& = -\underline{\sigma}'  t	 - \omega \lambda t, \quad \text{where} \quad \omega= \int_{\mathbb{R}^{+}} 
\left(e^{\rho x}-1 \right) \nu_Z (dx).
\end{align*}
For the last line, we use $\omega \leq 0$, as $\rho \leq 0$ and $Z$ is a subordinator. Similarly,
\begin{align*}
\left| \int_{0}^{t}\int_{\mathbb{R}^{+}} \left(e^{\rho x}-1 
\right)  \tilde{J}_Z(\lambda du, dx) \right| & = \left|\int_{0}^{t}\int_{\mathbb{R}^{+}} \left(e^{\rho x}-1 
\right) J_Z(\lambda du, dx)
- \int_{0}^{t}\int_{\mathbb{R}^{+}} \left(e^{\rho x}-1 
\right)  \nu_Z (dx) \lambda du \right| \\
& \geq  \left|\int_{0}^{t}\int_{\mathbb{R}^{+}} \left(e^{\rho x}-1 
\right) J_Z(\lambda du, dx)\right| - \left|\int_{0}^{t}\int_{\mathbb{R}^{+}} \left(e^{\rho x}-1 
\right)  \nu_Z (dx) \lambda du \right| \\
& \geq  -\int_{0}^{t}  \bar{\sigma}' du - \left|\int_{0}^{t} \lambda du 
\int_{\mathbb{R}^{+}} \left(e^{\rho x}-1 \right) \nu_Z 
(dx) \right| \\
& = -\bar{\sigma}'  t	 + \omega \lambda t.
\end{align*}
Hence we have
\begin{equation}
\label{kbkbkb}
 \left(-\bar{\sigma}'  t+ \omega \lambda t\right)^+ \leq \left| \int_{0}^{t}\int_{\mathbb{R}^{+}} \left(e^{\rho x}-1 
\right)  \tilde{J}_Z(\lambda du, dx) \right| \leq -\underline{\sigma}'  t - \omega \lambda t.
\end{equation}
\end{remark}

The price of the Asian call and put options with maturity $T$ and strike $K$ are 
given by
\begin{equation}
\label{sec25}
C(T) := e^{-rT} \mathbb{E}\left[ \left( \displaystyle\frac{1}{T} \int_0^{T}S_{t}dt - K \right) ^{+}\right], 
\end{equation}
\begin{equation}
	\label{sec26}
	P(T) := e^{-rT} \mathbb{E}\left[ \left(  K - \displaystyle\frac{1}{T} 
	\int_0^{T} S_{t}\, dt \right) ^{+}\right], 
\end{equation}
where $\mathbb{E}$ is taken with respect to the measure $\mathbb{P}$, and \emph{$C(T)$ and $P(T)$ emphasize the dependence on the maturity $T$}.

When $S_0 < K$, it is to be understood that the call option is out-of-the-money (OTM) and $C(T) \to 0$ as $T \to 0$ and when $S_0 > K$, we say that the put option is out-of-the-money (OTM) and $P(T) \to 0$ as $T \to 0$. On the other hand, when $S_0 > K$, the call option is in-the-money, and when $K > S_0$, the put option is in-the-money.  When $S_0 = K$, we say that the call and put options are 
at-the-money (ATM) and both $C(T)$ and $P(T)$ approach zero as $T \to 0$. In this paper, we investigate the first-order approximations of the call and put prices as $T \to 0$. Our goal is to show that the asymptotics for the out-of-the-money case are governed by the rare events also known as large deviations and the asymptotics for the ATM case are governed by the 
fluctuations about some other typical events.

We will use large deviations theory for small time diffusion processes. We will observe that application of the contraction principle in the 
field of large deviation theory  will be instrumental to obtain the corresponding large deviations for the small time arithmetic 
average of the diffusion, $\frac{1}{T}\int_{0}^{T} S_t dt$, and hence obtain 
the asymptotic behavior for the out-of-the-money Asian call and put options.  The asymptotic exponent is given as the rate function for the out-of-the-money 
Asian call and put options, from the large deviation principle. The asymptotics for in-the-money case is solved by 
applying the 
put-call parity. 


\section{Asymptotics for short maturity Asian options}
\label{sec4}
With respect to the filtered probability space $(\Omega, \mathcal{F}, (\mathcal{F}_t)_{0 \leq t \leq T}, \mathbb{P})$, the stochastic differential equation of the underlying asset $S = (S_t)_{t \geq 0}$ satisfying Assumption \ref{as512}, follows the dynamics \eqref{ch5sec21}.
We are interested in the short maturity limits, i.e., the asymptotics as $T \to 0.$

\subsection{Pricing of the short maturity Asian options}
\label{sec31}
\begin{lemma}
\label{lemma521}
Under the risk-neutral probability measure $\mathbb{P}$, the expectation of averaged asset price is given by 
\begin{equation}
A(T) := 	\frac{1}{T}\int_{0}^{T} \mathbb{E}[S_t] dt = 
            \frac{S_0}{T[(r-q) + \frac{\sigma^2}{2} +C]} \left[e^{T[(r-q) + \frac{\sigma^2}{2} +C]} -1 
\right], 
\label{lemma52111}
\end{equation}
where $r \geq 0$ is the risk-free interest rate, $q \geq 0$ is 
the continuous dividend yield, $\sigma$ satisfies Assumption \ref{as512}, $C= \lambda  \int_{\mathbb{R}+}(e^{\rho x} - 1 -\rho x 1_{|x| \leq 1}) \nu_Z(dx)$,
$\rho \leq 0$, $\lambda> 0$, and $\nu_Z $ is the 
L\'evy measure of the subordinator $Z$.
\end{lemma}
\begin{proof}
The underlying stock price $S_t$, given by $S_t = S_0 e^{X_t}$, $S_0 > 0$, satisfies the SDE \eqref{ch5sec21}, under the risk-neutral probability
measure $\mathbb{P}$, where $X_t = (r-q) t + \sigma W_t + \rho Z_{\lambda t}$. To find an expression for $A(T)$, we observe that
\begin{equation}
\label{sec3new31}
\mathbb{E}[S_t] = \mathbb{E}[S_0 e^{X_t}] = S_0 \mathbb{E}[e^{(r-q) t + 
\sigma W_t + \rho Z_{\lambda t}}] = S_0 e^{(r-q) t} 
\mathbb{E}[e^{\sigma W_t}] \mathbb{E}[e^{ \rho Z_{\lambda t}}],
\end{equation}
where the expectations in the last equation can be written as a product of 
expectations since $W_t \, \mathrm{and} \, Z_{\lambda t}$ are independent of 
each other. Also, $\mathbb{E}[e^{\sigma W_t}] = e^{\frac{t \sigma^2}{2}}$.  Next by applying L\'evy-Khinchin representation (\cite{cont}, Theorem 3.1), we obtain:
\begin{equation}
\label{sec3new33}
\mathbb{E} [e^{\rho Z_{\lambda t}}] = \exp\left(\lambda t \int_{\mathbb{R}+}(e^{\rho x} - 1 -\rho x 1_{|x| \leq 1}) \nu_Z(dx)\right). 
\end{equation}
Let us denote $C= \lambda  \int_{\mathbb{R}+}(e^{\rho x} - 1 -\rho x 1_{|x| \leq 1}) \nu_Z(dx)$, for some constant $C$. Substituting \eqref{sec3new33} back into \eqref{sec3new31} we get $\mathbb{E}[S_t] = S_0e^{t [(r-q) + \frac{\sigma^2}{2} + C]}$.
Consequently, we have
\begin{align*}
A(T):= \frac{1}{T}\int_{0}^{T} \mathbb{E}[S_t] dt = \frac{1}{T}\int_{0}^{T}  
S_0 e^{t [(r-q) + \frac{\sigma^2}{2} + C]} dt.
\end{align*}
From this \eqref{lemma52111} is obtained. 
\end{proof}
The Landau's symbol $O$ is understood as: for any arbitrary function $F \in O(T),$ for at least one choice of a constant $m > 0$, we can find a constant $n$ such that the inequality $0\leq F(T) \leq m T$ holds for all $T > n$. In other words, $F(T)$'s asymptotic growth is no faster than $T$'s.

On the other hand, if $g$ is any arbitrary function, then we say that $g \in 
o(\frac{1}{T})$ provided $\frac{g}{1/T} \to 0$. In other words, $g$'s 
asymptotic growth is strictly slower than $\frac{1}{T}$'s, where $o$ is the 
Landau's symbol.

\begin{remark}
We know that prices of call and put Asian options are related by put-call parity as
\begin{equation}
\label{3new34}
C(K,T) - P(K,T) = e^{-rT}(A(T) - K).
\end{equation}
Notice that when $r-q \neq 0,$
$$A(T) = \frac{S_0}{T[(r-q) + \frac{\sigma^2}{2} +C]} \left[e^{T[(r-q) + \frac{\sigma^2}{2} +C]} -1 \right],$$
whose first order Taylor series approximation is given by,
\begin{align*}
A(T)= \frac{S_0}{T[(r-q) + \frac{\sigma^2}{2} +C]} \left[ 1 + T[(r-q) + \frac{\sigma^2}{2} +C] + O(T^2) -1 \right],
\end{align*} 
which implies that as $ T \to 0, \,   A(T) = S_0 + O(T)$.
\end{remark}


\subsection{Short maturity out-of-the-money and in-the-money Asian options}
\label{sec32}
We will use large deviation theory to compute the leading-order approximation 
at $T \to 0$ for the price of the out-of-the-money Asian options.
\begin{theorem}
\label{thm 2}
Suppose that the Assumptions \ref{as512}, \ref{ch5as5224}, and \ref{asmp528} 
hold true. Then,
\begin{enumerate}
\item[(i)]
for out-of-the-money call Asian options, i.e., $K > S_0$,
\begin{equation}
\label{3new3}
C(T)= e^{-\frac{1}{T} \mathcal{I}(K, S_0) + o(\frac{1}{T})} \quad as  \quad T 
\to 0.
\end{equation}	
\item[(ii)] For out-of-the-money put Asian options, i.e., $K < S_0$,
\begin{equation}
\label{3new4}
P(T)= e^{-\frac{1}{T} \mathcal{I}(K, S_0) + o(\frac{1}{T})} \quad as  \quad T 
\to 0,
\end{equation}
\end{enumerate}
where for any $S_0, K >0$, and $\mathcal{I}(K, S_0) $ is known as the rate function.


\end{theorem}

\begin{proof}
(i) We start by first proving the relation
\begin{equation}
\label{sec338new}
\lim_{T \to 0} T \log C(T) = \lim_{T \to 0} T \log \mathbb{P} \left( \frac{1}{T}\int_{0}^{T} S_t dt \geq K \right).
\end{equation}
Recall that $C(T) = e^{-rT} \mathbb{E} \left[ \left( \frac{1}{T}\int_{0}^{T} 
S_t dt - K \right)^{+} \right] = e^{-rT} \mathbb{E} \left[ \max \left(\left( 
\frac{1}{T}\int_{0}^{T} S_t dt - K \right) , 0 \right)  \right] $. By 
H\"older's inequality, for any $\frac{1}{p} + \frac{1}{q} = 1, \quad p,	q > 
1$,
\begin{align*}
\label{sec338}
C(T) & \leq e^{-rT} \mathbb{E} \left[ \left| \frac{1}{T}\int_{0}^{T} S_tdt - 
K \right| \mathcal{X}_{	\frac{1}{T}\int_{0}^{T} S_t dt \geq K} \right] \nonumber \\
& \leq e^{-rT} \left(\mathbb{E} \left[ \left| \frac{1}{T}\int_{0}^{T} S_tdt- 
K \right| ^p \right] \right)^{\frac{1}{p}} \mathbb{P} \left(\frac{1}{T}\int_{0}^{T} S_t dt \geq K \right) ^{\frac{1}{q}}.
\end{align*}
Assume that $p \geq 2$. Note that for $p \geq 2, \, x \mapsto x^p$ is a convex function for $x \geq 0$ and by Jensen's inequality, $\left( \frac{x+y}{2} \right)^p \leq \frac{x^p + y^p}{2}$ for any $x, y \geq 0$, 
\begin{equation}
\label{sec339}
\mathbb{E} \left[ \left| \frac{1}{T}\int_{0}^{T} S_t dt	- K \right| ^p \right]
\leq \mathbb{E} \left[ \left( \frac{1}{T}\int_{0}^{T} S_t dt +	K \right)^p \right]
\leq 2^{p-1} \left[ \mathbb{E} \left[ \left( \frac{1}{T}\int_{0}^{T}S_t dt\right)^p \right] +	K^p  \right].
\end{equation}
By Jensen's inequality again, 
\begin{equation}
		\label{sec3310}
		\mathbb{E} \left[ \left( \frac{1}{T}\int_{0}^{T} S_t dt \right)^p 
		\right] 
		\leq \mathbb{E}\left[  \frac{1}{T}\int_{0}^{T} S_t^{p} dt \right]
		\leq \frac{1}{T}\int_{0}^{T} \mathbb{E}[S_t^{p}]  dt. 
\end{equation}
Recall that the dynamics of asset price $S_t$ under the risk-neutral probability measure $\mathbb{P}$ is given by
\begin{align*}
\frac{dS_t}{S_t} = (r-q + \frac{1}{2}\sigma^{2})dt + \sigma dW_t + \int_{\mathbb{R}^{+}}(e^{\rho x}-1){J}_{Z}(\lambda dt, dx), \quad S_0 > 0.
\end{align*}
By It\^o's formula,
\begin{align*}
d(S_t^{p}) & = p S_t^{p-1} dS_t + \frac{1}{2} p(p-1) S_t^{p-2} (dS_t)^2 \\
& = S_t^{p} \left[\left(p(r-q) + \frac{p^2}{2}\sigma^{2}\right) dt + p\sigma dW_t 
+ p\int_{\mathbb{R}^{+}}(e^{\rho x}-1){J}_{Z}(\lambda dt, dx) 
\right]	\\
& + \frac{1}{2} p(p-1) S_t^{p} (\int_{\mathbb{R}^{+}}(e^{\rho x}-1){J}_{Z}(\lambda dt, dx)) ^2.
\end{align*}
Therefore, 
\begin{align*}
d\mathbb{E}[S_t^{p}] = p (r-q) \mathbb{E}\left[S_t^{p} \right]dt + 
\frac{p^2}{2} \mathbb{E}\left[\sigma^{2} S_t^{p} \right]dt + p\mathbb{E}\left[S_t^{p} \int_{\mathbb{R}^{+}}(e^{\rho x}-1){J}_{Z}(\lambda dt, dx) \right]\\
 + \frac{1}{2} p(p-1) \mathbb{E}\left[S_t^{p} 
 (\int_{\mathbb{R}^{+}}(e^{\rho x}-1){J}_{Z}(\lambda dt, dx)) ^2 \right].
\end{align*}	
We conclude that $\mathbb{E}[S_t^{p}] \leq a(t)$, where $a(t)$ is the solution to the ODE 
\begin{align*}
d a(t) = a(t)\left[p (r-q + \bar{\sigma}') + \frac{p^2}{2} \sigma^2 + 
\frac{1}{2} p(p-1)\bar{\mu} \right]dt, \quad a(0) = S_0^{p}.
\end{align*}
Thus $a(t) = S_0^{p} e^{(p (r-q + \bar{\sigma}') + \frac{p^2}{2}\sigma^2 +	\frac{1}{2} p(p-1)\bar{\mu} t}$. Hence, 
\begin{equation}
\label{sec3315}
\frac{1}{T}\int_{0}^{T} \mathbb{E}[S_t^{p}]  dt \leq \max_{0 \leq t \leq T} a(t)
\leq S_0^{p} e^{|e^{(p (r-q + \bar{\sigma}') + \frac{p^2}{2} \sigma^2 + \frac{1}{2}p(p-1)\bar{\mu} t}|T}.
\end{equation}
Therefore, by \eqref{sec339}, \eqref{sec3310}, and \eqref{sec3315}, we have
\begin{equation*}
\label{sec3316}
\limsup_{T \to 0} T\log C(T) \leq \limsup_{T \to 0} \frac{1}{q} T \log 
\mathbb{P} \left( \frac{1}{T}\int_{0}^{T} S_t dt \geq K \right).
\end{equation*}
Since this holds for any $1 < q \leq 2$, we have the upper bound. For any $\epsilon > 0$, 
\begin{equation*}
\label{sec3317}
C(T) \geq e^{-rT} \mathbb{E} \left[ \left( \frac{1}{T}\int_{0}^{T} S_t dt - K \right) \mathcal{X} _ {\frac{1}{T}\int_{0}^{T} S_t dt \geq K + \epsilon}  \right]
\end{equation*}
$$ \geq e^{-rT}\epsilon \mathbb{P} \left(\frac{1}{T}\int_{0}^{T} S_t dt 
\geq K + \epsilon \right).$$
Consequently,
\begin{equation*}
	\label{sec3318}
	\liminf_{T \to 0} T\log C(T) \geq \liminf_{T \to 0} T \log \mathbb{P} 
	\left(\frac{1}{T}\int_{0}^{T} S_t dt \geq K + \epsilon \right).
\end{equation*}
Since this holds for any $\epsilon > 0$, we get the lower bound.
We are interested to computing the limit of $\lim_{T \to 0} T \log \mathbb{P} \left(\frac{1}{T}\int_{0}^{T} S_t dt \geq K \right)$. Observe,
\begin{equation*}
\label{sec3319}
\lim_{T \to 0} T \log \mathbb{P} \left(\frac{1}{T}\int_{0}^{T} S_t dt \geq K \right) = 
\lim_{T \to 0} T \log \mathbb{P} \left(\frac{1}{T}\int_{0}^{1} S_{tT} dt \geq K \right).
\end{equation*}
Let $X_t := \log S_t$. This is equivalent to computing the limit
\begin{equation*}
\label{sec3320}
\lim_{T \to 0} T \log \mathbb{P} \left(\frac{1}{T}\int_{0}^{1} e^{X_{tT}} dt \geq K \right),
\end{equation*}
where by It\^o's lemma, we can write
\begin{equation*}
\label{sec3321}
dX_t = (r - q) dt + \sigma dW_t + \rho_1 \kappa dZ_{\frac{t}{\kappa}},  \quad X_0 = \log S_0,
\end{equation*}
where the parameters $\rho_1, \kappa, \lambda \in \mathbb{R}$ and we define 
$\rho := \rho_1 \kappa$, for $\rho_1 \leq 0, \, \kappa > 0$ and $\lambda = 
\frac{1}{\kappa} > 0.$	

From the large deviations theory for small time diffusions, it was first proved 
in Varadhan \cite{vara} that under the assumptions \eqref{sec23} and 
\eqref{sec24new} (i.e., when $\sigma(\cdot)$ is considered for the local volatility model), $\mathbb{P}(X_{\cdot T} \in \cdot)$ satisfies a sample path 
large deviation principle on $L_{\infty} [0,1]$. In particular, with Assumption \ref{as512} this will be satisfied. In our case, with a similar analysis done 
in Chapter 7 of \cite{kuhn} and by an application of the contraction principle (see 
Theorem \ref{thm22}), we find that $\mathbb{P} \left(\int_{0}^{1} 
e^{X_{tT}} dt \in \cdot \right)$ satisfies a large deviation principle with the 
rate function $\mathcal{I}(x, S_0)$ (the notation is as explained in \eqref{koya}).
Hence, for out-of-the-money call options i.e., $S_0 < K$,
\begin{equation*}
\label{sec3323}
\lim_{T \to 0} T \log \mathbb{P} \left(\frac{1}{T}\int_{0}^{T} S_t 
dt \geq K \right) = - \inf_{a \geq K} 	\mathcal{I}(a, S_0) = - \mathcal{I}(K,S_0),
\end{equation*}
where the last step is due to the fact that $\mathcal{I}(K, S_0)$ is increasing 
in $K$ for $K > S_0$.  Our goal is to show that $	C(T)= e^{-\frac{1}{T} \mathcal{I}((K, S_0)) + o(\frac{1}{T})} \quad \mathrm{as} \quad T \to 0$. Hence, by first proving the relation that 
\begin{align*}
\lim_{T \to 0} T \log C(T) = \lim_{T \to 0} T \log \mathbb{P} \left( \frac{1}{T}\int_{0}^{T} S_t dt \geq K \right),
\end{align*}
and applying large deviation principle we found the limiting value of $$\lim_{T 
\to 0} T \log \mathbb{P} \left(\frac{1}{T}\int_{0}^{T} S_t dt \geq K \right) = - \mathcal{I}(K,S_0).$$
Consequently, it is clear  that for out-of-the-money call option, $C(T) =  e^{-\frac{1}{T} \mathcal{I}(K, 
S_0) + o(\frac{1}{T})}$,  as $T \to 0$.
\end{proof}
We can prove an analogous relation to \eqref{sec338new} for out-of-the-money 
Asian put options, $S_0 > K$. We omit the proof. 
\begin{remark}
The small maturity asymptotics given by Theorem \ref{thm 2} and the rate function $\mathcal{I}(K,S_0)$ are independent of the interest rate $r$ and dividend yield $q$.
\end{remark}
The asymptotics for short maturity in-the-money Asian call and put options can 
be obtained by the application of the put-call parity.
\begin{theorem}
\label{cor 1}
Suppose that Assumptions \ref{as512}, \ref{ch5as5224}, and \ref{asmp528}  hold.\\
(i) For in-the-money call Asian options, i.e. $K < S_0,$
\begin{equation}
\label{3new326}
C(T)= S_0- K +r KT + \frac{S_0 T (\frac{\sigma^2}{2} + C -q)^2 
}{2[(r-q) + \frac{\sigma^2}{2} + C]} + O(T^2), \, \mathrm{as}  \quad T \to 0.
\end{equation}
(ii) For in-the-money put Asian options, i.e., $K > S_0$, 
\begin{equation}
\label{3new327}
P(T)= K - S_0  - r KT - \frac{S_0 T (\frac{\sigma^2}{2} + 
C -q)^2 }{2[(r-q) + \frac{\sigma^2}{2} + C]}  + O(T^2), \, \mathrm{as} \quad T 
\to 0.
\end{equation}
\end{theorem}

\begin{proof} 
\begin{enumerate}
\item[(i)]
From the put-call parity, 
\begin{align*}
C(T)-P(T) & =	e^{-rT}\mathbb{E}\left[\frac{1}{T} \int_{0}^{T} S_t dt - K \right] \\
& = e^{-rT}\left[\frac{1}{T} \int_{0}^{T} \mathbb{E}[S_t]dt - K\right],  \quad	\mathrm{where} \, S_t = S_0 e^{X_t}. 
\end{align*}
From Lemma \ref{lemma521} we obtain $\mathbb{E}[S_t]= S_0 e^{t [(r-q) + 
\frac{\sigma^2}{2} + C]} $, where $C= \lambda  \int_{\mathbb{R}+}(e^{\rho x} - 1 -\rho x 1_{|x| \leq 1}) \nu_Z(dx)$. We then have 
\begin{align*}
C(T) -P(T) & =e^{-rT} \left[ \frac{1}{T} \int_{0}^{T} S_0 e^{t [(r-q) + 
\frac{\sigma^2}{2} + C]} dt -K \right] \\
& = e^{-rT} \left[ S_0\frac{(e^{((r-q) + \frac{\sigma^2}{2} + C)T}-1)}{T[(r-q) + \frac{\sigma^2}{2} + C]} -K \right] .
\end{align*}

Notice that, when $r \neq q$, by applying Taylor series expansion for the exponential function, we can rewrite the expression as
\begin{equation*}
\frac{S_0 [1 + T( \frac{\sigma^2}{2} + C -q) + \frac{1}{2}(\frac{\sigma^2}{2} + 
C -q)^2 T^2 - (1- rT) + O(T^2)]}{T[(r-q) + \frac{\sigma^2}{2} + C ]} - (K -r 
KT + O(T^2)).
\end{equation*}
which can be further simplified to
\begin{align*}
= \frac{S_0 T ((r-q) + \frac{\sigma^2}{2} + C)}{T[(r-q) + \frac{\sigma^2}{2} + 
C ]} + \frac{S_0 (\frac{\sigma^2}{2} + C-q)^2 T^2}{2T[(r-q) + 
\frac{\sigma^2}{2} + C ]} - K +r KT + O(T^2)\\
= S_0 - K +r KT + \frac{S_0 T (\frac{\sigma^2}{2} + C -q)^2 }{2[(r-q) 
+ \frac{\sigma^2}{2} + C]} + O(T^2), \quad \mathrm{if} \quad r \neq q.
\end{align*}
Similarly we can obtain an expression for $r = q$.  Hence, 
\begin{equation*}
C(T)- P(T) = S_0- K +r KT + \frac{S_0 T (\frac{\sigma^2}{2} + C -q)^2 }{2[(r-q) + 
\frac{\sigma^2}{2} + C]} + O(T^2), 
\end{equation*}
as $T \to 0$. For an in-the-money call option, i.e., $S_0 > K$, we know from Theorem \ref{thm 2} that the put option is out-of-the-money and is given by $P(T)= e^{-\frac{1}{T} \mathcal{I}(K, S_0) + o(\frac{1}{T})}$. Since the asymptotic growth of the exponential function for an out-of-the-money put option is slower than $T^2, \, \mathrm{as} \, T \to 0$, therefore, we get $$C(T) = S_0 - K +r KT + \frac{S_0 T 
(\frac{\sigma^2}{2} + C -q)^2 }{2[(r-q) + \frac{\sigma^2}{2} + C]} + O(T^2), \, 
\mathrm{as}  \quad T \to 0.$$

\item[(ii)] For in-the-money put option, i.e., $S_0 < K$, from part (i) and 
from Theorem \ref{thm 2} we know that the call option is out-of-the-money given 
by $C(T)= e^{-\frac{1}{T} \mathcal{I}(K, S_0) + o(\frac{1}{T})}$. With a 
similar argument as in (i) we obtain $$P(T) = K - S_0 - r KT - \frac{S_0 
T (\frac{\sigma^2}{2} + C -q)^2 }{2[(r-q) + \frac{\sigma^2}{2} + C]}  + O(T^2), 
\, \mathrm{as} \quad T \to 0.$$	
\end{enumerate}
\end{proof}
\subsection{Short maturity at-the-money Asian options}
\label{sec33}
When $K = S_0$, the Asian call and put options are ATM. Let $Y_t = e^{-(r-q)t} S_t$. We assume that Assumptions \ref{as512} and \ref{asmp528} hold true. With respect to the risk-neutral measure $\mathbb{P}$, the process $Y_t = e^{-(r-q)t} S_t$ is a martingale.

Before proving the main result, we will prove some results that will be instrumental.
\begin{lemma}
\label{Claim 1}
As $T \to 0$, (with respect to risk-neutral measure $\mathbb{P}$),
\begin{equation}
\label{thm34eqn2}
\left | \mathbb{E} \left[ \left( \frac{1}{T}\int_{0}^{T} e^{(r-q)t} Y_t dt - S_0 \right)^{+} \right] - \mathbb{E} \left[ \left( \frac{1}{T}\int_{0}^{T} Y_t dt - S_0 \right)^{+} \right]	\right | = O(T).
\end{equation}
\end{lemma}
\begin{proof}
We prove \eqref{thm34eqn2} by rewriting
\begin{equation*}
\label{thm34neweqn2}
\left | \mathbb{E}  \left[ \left( \frac{1}{T}\int_{0}^{T} e^{(r-q)t} Y_t dt - S_0 \right)^{+} \right] - \mathbb{E} \left[ \left( \frac{1}{T}\int_{0}^{T} Y_t dt  - S_0 \right)^{+} \right] \right |
\end{equation*}
\begin{align*}
\leq \mathbb{E}  \left[ \left | \left( \frac{1}{T}\int_{0}^{T} e^{(r-q)t} Y_t dt - S_0 \right)^{+} - \left( \frac{1}{T}\int_{0}^{T} Y_t dt - S_0 \right)^{+}  \right | \right].
\end{align*}
Since $\left( \frac{1}{T}\int_{0}^{T} e^{(r-q)t} Y_t dt - S_0 \right)^{+} = 
\max(\frac{1}{T}\int_{0}^{T} e^{(r-q)t} Y_t dt - S_0, 0)$ and $\left( 
\frac{1}{T}\int_{0}^{T} Y_t dt - S_0 \right)^{+} = \max(\frac{1}{T}\int_{0}^{T} 
Y_t dt - S_0, 0) $, the above inequality is
\begin{align*}
\leq \mathbb{E} \left[ \left | \frac{1}{T}\int_{0}^{T} e^{(r-q)t} Y_t dt - S_0  
-  \frac{1}{T}\int_{0}^{T} Y_t dt + S_0\right | \right] \leq 	\mathbb{E} 
\left[ \frac{1}{T}\int_{0}^{T} \left |	
e^{(r-q)t} -1 \right | Y_t dt   \right].
\end{align*}
Since $Y_t$ is a martingale observe that after taking expectation inside the 
integral the integrand becomes
\begin{align*}	
\left |e^{(r-q)t} -1 \right | \mathbb{E} [Y_t] = \left |e^{(r-q)t} -1 \right | Y_0 = S_0 \left |e^{(r-q)t} -1 \right |.
\end{align*}
This gives us
\begin{align*}
\mathbb{E} \left[ \frac{1}{T}\int_{0}^{T} \left | e^{(r-q)t} -1 \right | Y_t dt  \right] 
= \frac{S_0}{T}\int_{0}^{T} \left | e^{(r-q)t} -1 \right | dt = S_0 \left 
|\frac{1}{T}\int_{0}^{T}  (e^{(r-q)t} -1 ) dt \right |.
\end{align*}
Integrating the last equality from $0 \, \mathrm{to}\,  T$ gives us $S_0 \left 
| \frac{e^{(r-q)T} -1}{(r-q)T} -1 \right | $. Applying the Taylor series approximation up to order two and upon further simplification we confirm that the difference  of the two averages can be approximated as $O(T)$. 
Hence, we confirmed \eqref{thm34eqn2}.
\end{proof}

Next, we define a martingale $\hat{Y}_t$ that satisfies the SDE
\begin{equation}
\label{thm34eqn3}
d\hat{Y}_t = \sigma S_0 d{W}_t + S_0 \int_{\mathbb{R}^{+}} 
\left(e^{\rho x}-1 \right) \tilde{J}_Z(\lambda dt, dx), \quad 
\mathrm{with} \quad \hat{Y}_0 = S_0,
\end{equation}	
under the risk-neutral measure $\mathbb{P}$.

\begin{theorem}
\label{Claim 2} 
Suppose that Assumption \ref{ch5as5224} holds true. Then, for $\hat{Y}_t $ as defined in \eqref{thm34eqn3}, 
\begin{equation}
\label{thm34eqn4}
\mathbb{E} \left[ \max_{0\leq t \leq T} \left| Y_t - \hat{Y}_t \right| \right] = O(T) \quad \mathrm{as} \quad  T\to 0.
\end{equation}
\end{theorem}
\begin{proof}
We prove \eqref{thm34eqn4} by observing that
\begin{equation*}
\label{thm34eqn5}
\int_{0}^{t} (dY_u - d\hat{Y}_u) = Y_t - Y_0 - \hat{Y}_t + \hat{Y}_0 = Y_t-\hat{Y}_t,
\end{equation*}
as $Y_0 = \hat{Y}_0 = S_0$.
Then,
\begin{align}
	\label{thm5eqn49}
Y_t - \hat{Y}_t & = \int_{0}^{t}\left( \sigma Y_u - \sigma 
S_0 \right) d W_u + \int_{0}^{t} (Y_u - S_0 ) \int_{\mathbb{R}^{+}} 
\left(e^{\rho x}-1 \right) \tilde{J}_Z(\lambda du, dx).
\end{align}
Squaring both sides of \eqref{thm5eqn49}  we obtain:
\begin{align*}
(Y_t - \hat{Y}_t)^2 & =\left(\int_{0}^{t}\left( \sigma Y_u - 
\sigma S_0 \right) d{W}_u \right)^2 + \int_{0}^{t} \int_{\mathbb{R}^{+}} 
(Y_u - S_0 )^2 \left(e^{\rho x}-1 \right) ^2 J_Z(\lambda du, 
dx),
\end{align*}
Thus by  It\^o's isometry:
\begin{align}
\label{thm5eqnito}
\mathbb{E} \left[ (Y_t - \hat{Y}_t)^2 \right] = \int_{0}^{t} \mathbb{E} \left[ 
\left( \sigma Y_u - \sigma S_0 \right)^2 \right]du +  
\mathbb{E} \left[ \int_{0}^{t} \int_{\mathbb{R}^{+}} (Y_u - S_0 )^2 
\left(e^{\rho x}-1 \right) ^2 J_Z(\lambda du, 
dx)  \right].
\end{align}
We note that the first term of \eqref{thm5eqnito} can be rewritten as:
\begin{align*}	
\int_{0}^{t}\mathbb{E} \left[ \left( \sigma Y_u - \sigma 
S_0 \right)^2\right]du &= \sigma^2 \int_{0}^{t}\mathbb{E} \left[ \left( 
Y_u - S_0 \right)^2\right]du
\end{align*}
\begin{equation}
\label{thm5eqnsig1}
\leq 2 \sigma^2 \int_{0}^{t}\mathbb{E} \left[ \left( Y_u- 
\hat{Y}_u\right)^2\right]du + 2 \sigma^2 \int_{0}^{t}\mathbb{E} \left[ 
\left(\hat{Y}_u-S_0 \right)^2\right]du.
\end{equation}
By Assumption \ref{ch5as5224} the second term of \eqref{thm5eqnito} is given by:
\begin{equation}
\label{thm5eqnasm}
\begin{split}
\mathbb{E} \left[ \int_{0}^{t} \int_{\mathbb{R}^{+}} (Y_u - S_0 )^2 
\left(e^{\rho x}-1 \right) ^2 J_Z(\lambda du, 
dx)  \right] \leq \bar{\mu} \int_{0}^{t} \mathbb{E} \left[(Y_u - S_0 )^2 
\right] du \\
\leq 2 \bar{\mu} \int_{0}^{t} \mathbb{E} \left[(Y_u - \hat{Y}_u )^2 \right] 
du + 2 \bar{\mu} \int_{0}^{t} \mathbb{E} \left[(\hat{Y}_u - S_0 )^2 \right] du.
\end{split}
\end{equation}
Substituting \eqref{thm5eqnsig1} and \eqref{thm5eqnasm} back into 
\eqref{thm5eqnito} we get:
\begin{equation}
\label{thm52eq51}
\begin{split}
\mathbb{E} \left[ (Y_t - \hat{Y}_t)^2 \right] \leq (2\sigma^2 + 2\bar{\mu}) 
\int_{0}^{t}\mathbb{E} \left[ \left( Y_u-\hat{Y}_u\right)^2\right] du +( 
2\sigma^2 + 2\bar{\mu}) \int_{0}^{t}\mathbb{E} \left[ \left(\hat{Y}_u-S_0 
\right)^2\right]du.
\end{split}
\end{equation}
To find explicit expressions for the terms on the right hand side 
of \eqref{thm52eq51}, we note that $$d \hat{Y}_u = \sigma S_0 d{W}_u + S_0 
\int_{\mathbb{R}^{+}} \left(e^{\rho x}-1 \right)  
\tilde{J}_Z(\lambda du, dx), \quad \mathrm{with} \quad 
\hat{Y}_0 = S_0.$$ 
Consequently,  $\hat{Y}_u = S_0 + \sigma S_0 {W}_u + S_0 
\int_{0}^{u}\int_{\mathbb{R}^{+}} \left(e^{\rho x}-1 \right)  
\tilde{J}_Z(\lambda ds, dx)$. 
Therefore 
\begin{align*}
(\hat{Y}_u - S_0)^2 = \sigma^2 S_0^2 {W_u}^2 + S_0^2 \left(
\int_{0}^{u}\int_{\mathbb{R}^{+}} \left(e^{\rho x}-1 \right)  
\tilde{J}_Z(\lambda ds, dx) \right)^2.
\end{align*}
Hence,
\begin{align*}
\mathbb{E} \left[ (\hat{Y}_u - S_0)^2 \right] = \sigma^2 S_0^2  u +
S_0^2 \, \mathbb{E} \left[\int_{0}^{u}\int_{\mathbb{R}^{+}} \left(e^{\rho x}-1 \right)^2 J_Z(\lambda ds, dx)\right]
\leq  \sigma^2 S_0^2  u + S_0^2  \int_{0}^{u} \mathbb{E} [\bar{\mu}]ds
\end{align*}
Thus, we have:
\begin{equation*}
\label{thm5eqn54}
\mathbb{E} \left[ (\hat{Y}_u - S_0)^2 \right]  \leq \sigma^2 S_0^2  u + 
S_0^2 \bar{\mu} u = S_0^2 (\sigma^2  + \bar{\mu})u.
\end{equation*}
It is clear that  
\begin{equation*}
\label{thm5eqnasmsigma}
|\sigma x|  \leq \alpha|x|, \quad \mathrm{such \, that} \, \sigma \leq \alpha.
\end{equation*}
Since $Y_t - Y_0$ is a martingale starting at 0, by the	
Burkholder-Davis-Gundy 
inequality, Theorem \ref{ch2thmburk}, for $p = 4$ we can write,
\begin{align}
\label{thm5eqnbdg}
\mathbb{E} \left[ ({Y}_t - Y_0)^4 \right] \leq C \mathbb{E} \left[ ([Y, Y]_t)^2 
\right] \, \mathrm{for \, some \, 
constant \,} C> 0.
\end{align}
We recall that 
\begin{equation*}
dY_u = \sigma Y_u d W_u + Y_u \int_{\mathbb{R}^{+}} 
\left(e^{\rho x}-1\right) \tilde{J}_Z(\lambda du, dx).
\end{equation*}
Consequently,
\begin{align*}
d[Y,Y]_{u} & = dY_u \cdot dY_u \\
& = \left(\sigma Y_u d W_u \right)^2 + \left(Y_u 
\int_{\mathbb{R}^{+}} \left(e^{\rho x}-1\right) 
\tilde{J}_Z(\lambda du, dx) \right)^2.
\end{align*}
Thus,
\begin{align*}
[Y, Y]_t & = \int_{0}^{t} d[Y,Y]_{u} = \int_{0}^{t}\sigma^2 Y^2_{u} du + \int_{0}^{t}\int_{\mathbb{R}^{+}} Y^2_{u} 
\left(e^{\rho x}-1\right)^2 {J}_Z(\lambda du, dx) \\
& \leq \int_{0}^{t}\sigma^2  Y^2_{u} du + \bar{\mu} \int_{0}^{t} 
Y^2_{u} du.
\end{align*}
Therefore \eqref{thm5eqnbdg} gives,
\begin{align*}
\mathbb{E} \left[ ({Y}_t - Y_0)^4 \right] \leq C \mathbb{E} \left[ 
\left(\int_{0}^{t}\sigma^2 Y^2_{u} du + \bar{\mu} \int_{0}^{t} 
Y^2_{u} du \right) ^2  \right]
\end{align*}
\begin{align}
\label{thm5eqnqcov}	
& \leq 2 C \mathbb{E} \left[ \left(\int_{0}^{t}\sigma^2 Y^2_{u} 
du \right)^2 \right] +  2 C \mathbb{E} \left[ \left(\int_{0}^{t} \bar{\mu}
Y^2_{u} du \right) ^2\right]. 
\end{align}
Using the Cauchy-Schwartz inequality on \eqref{thm5eqnqcov}, we obtain:
\begin{align*}
\mathbb{E} \left[ ({Y}_t - Y_0)^4 \right] & \leq 2C \mathbb{E} \left[ 
\int_{0}^{t}\sigma^4 du \int_{0}^{t}Y^4_{u}du \right] + 2C 
\mathbb{E} \left[\int_{0}^{t} \bar{\mu}^2 du \int_{0}^{t} Y^4_{u} du \right].
\end{align*}	
Since $\sigma \leq \alpha$ and $(\frac{x+y}{2})^4 \leq 
\frac{x^4+y^4}{2}$ for any $x, y \geq 0$, considering $x = Y_0, \, y = Y_t - 
Y_0$, we get $$ \mathbb{E} \left[{Y_t}^4\right] \leq  8Y^4_0 + 8\mathbb{E} 
\left[ ({Y}_t - Y_0)^4 \right], \quad Y_0 = S_0.$$ 
Therefore, for sufficiently small $t$, say $t \leq 1$,
\begin{align*}
\mathbb{E} \left[{Y_t}^4\right]  & \leq  8S^4_0 + 8\left( 2C \alpha^4  
\int_{0}^{t}\mathbb{E} [Y^4_{u}]du  + 2C \bar{\mu}^2 \int_{0}^{t} \mathbb{E} 
[Y^4_{u}]du \right)\\
& \leq 8S^4_0 + A \int_{0}^{t}\mathbb{E} [Y^4_{u}]du, \quad A = 16C \alpha^4 + 
16 C  \bar{\mu}^2.
\end{align*}
Gronwall's inequality states that if $\gamma(\cdot)$ is non-negative and 
for any $t \geq 0, u(t) \leq \eta(t) + \int_{0}^{t}\gamma(s)u(s)ds$, then 
\begin{equation*}
\label{thm34eqn49}
u(t) \leq \eta(t) + \int_{0}^{t} \eta(s) \gamma(s) e^{\int_{s}^{t} 
\gamma(r)dr} ds.
\end{equation*}
Therefore, by Gronwall's inequality we have for any sufficiently small 
 $t,$
$$\mathbb{E} \left[Y^4_t\right] \leq 8S^4_0 + \int_{0}^{t} 8S^4_0 A 
e^{\int_{s}^{t} A dr} ds = 8S^4_{0} e^{At}.$$
Hence, there exists a constant $m > 0$ so that,
\begin{equation*}
\label{thm5eqnbeta}
2\beta^2 \int_{0}^{t} \left(e^{(r-q)u}-1\right)^2\mathbb{E} [Y^4_u]du
\leq 2\beta^2 8S^4_{0} e^{At} \int_{0}^{t} \left(e^{(r-q)u}-1\right)^2 du
\leq mt^2,
\end{equation*}
for sufficiently small $t > 0$. Substituting \eqref{thm5eqn54} back into 
\eqref{thm52eq51} we get
\begin{align*}
\mathbb{E} \left[ (Y_t - \hat{Y}_t)^2 \right] & \leq (2\sigma^2 + 2\bar{\mu}) 
\int_{0}^{t}\mathbb{E} \left[ \left( Y_u-\hat{Y}_u\right)^2\right] du +( 
2\sigma^2 + 2\bar{\mu}) \int_{0}^{t} S_0^2(\sigma^2  + \bar{\mu})u du \\
& \leq (2\sigma^2 + 2\bar{\mu}) \int_{0}^{t}\mathbb{E} \left[ \left( 
Y_u-\hat{Y}_u\right)^2\right] du + S_0^2 (\sigma^2 + \bar{\mu})^2 t^2.
\end{align*}
Let,
\begin{align*}
\eta(t) = S_0^2 (\sigma^2 + \bar{\mu})^2 t^2 , \quad \gamma(s) = (2\sigma^2 + 
2\bar{\mu}). 	
\end{align*}
Therefore, by Gronwall's inequality we have for any sufficiently small $t$, we 
then have
\begin{align}
\label{thm5eqngronw2}
\mathbb{E} \left[ (Y_t - \hat{Y}_t)^2 \right] & \leq	
S_0^2 (\sigma^2 + \bar{\mu})^2 t^2  +  \int_{0}^{t} S_0^2 (\sigma^2 + 
\bar{\mu})^2 (2\sigma^2 + 2\bar{\mu}) s^2 e^{(2\sigma^2 + 2\bar{\mu})(t - s)} 
ds.
\end{align}
Applying integration by parts twice we can find an explicit expression for the 
right side of \eqref{thm5eqngronw2}. We conclude that there 
exists some universal constant $M' > 0$, such that for sufficiently small 
$T>0$, $\mathbb{E} \left[ (Y_T - \hat{Y}_T)^2 \right] \leq M' T^2$.

Finally we note that $Y_t -\hat{Y}_t$ is a martingale since both $Y_t$ and 
$\hat{Y}_t$ are. By Doob's martingale inequality, for sufficiently small $T > 
0,$\\ 
\begin{equation*}
\label{thm34eqn53}
\mathbb{E} \left[\max_{0\leq t \leq T} |Y_t - \hat{Y}_t| \right] 
\leq 2\left( \mathbb{E} \left[ (Y_T- \hat{Y}_T)^2 \right] 
\right)^{\frac{1}{2}} \leq 2 \sqrt{M'}T.
\end{equation*}
Hence the theorem is proved. 
\end{proof}

\begin{lemma}
\label{Claim 3}
Suppose that Assumption \ref{ch5as5224} holds true, and $\hat{Y}_t $ is defined by \eqref{thm34eqn3}. Then, for $T > 0$
\begin{equation}
\label{thm34eqn54}
\left|\mathbb{E} \left[ \left( \frac{1}{T}\int_{0}^{T} Y_t dt - S_0 \right)^{+} \right] - \mathbb{E}  \left[ \left( \frac{1}{T}\int_{0}^{T} \hat{Y}_t dt - S_0 \right)^{+} \right]\right| = O(T).
\end{equation}
\end{lemma}
\begin{proof}
\begin{align*}
\label{thm34eqn55}
& \left|\mathbb{E}  \left[ \max\left(\frac{1}{T}\int_{0}^{T} Y_t dt - S_0, 
0\right) \right]- \mathbb{E} \left[ \max\left(\frac{1}{T}\int_{0}^{T} \hat{Y}_t 
dt - S_0, 0\right) \right] \right| \\
& \leq \left|\mathbb{E} \left[ \frac{1}{T}\int_{0}^{T} Y_t dt - 
\frac{1}{T}\int_{0}^{T} \hat{Y}_t dt \right]  
\right| \\
& \leq \mathbb{E} \left[  \left| \frac{1}{T} \left( \int_{0}^{T} Y_t 
dt-\int_{0}^{T} \hat{Y}_t dt \right) \right| \right] \\
& \leq \mathbb{E} \left[ \frac{1}{T} \int_{0}^{T} |Y_t - \hat{Y}_t|dt \right]\\
& \leq \mathbb{E} \left[  \max_{0\leq t \leq T}|Y_t - \hat{Y}_t| \frac{1}{T} 
\int_{0}^{T} dt \right]
= \mathbb{E} \left[  \max_{0\leq t \leq T}|Y_t -\hat{Y}_t| \right].
\end{align*}
Hence, \eqref{thm34eqn54} follows from Theorem \ref{Claim 2}.
\end{proof}

\begin{theorem}
Suppose that Assumptions \ref{ch5as5224} holds true, and $\hat{Y}_t $ is defined by \eqref{thm34eqn3}. Then
\label{Claim 4}
\begin{equation}
\label{ayo}
\left(\frac{\sigma S_0 }{\sqrt{6 \pi}}\sqrt{T} - \frac{S_0 T}{4}(-\underline{\sigma}'  - \omega \lambda )\right)^+ \leq 
\mathbb{E} \left[ \left( \frac{1}{T}\int_{0}^{T} \hat{Y}_t dt - S_0 \right)^{+} 
\right] \leq  \frac{\sigma S_0 }{\sqrt{6 \pi}}\sqrt{T} + \frac{S_0 T}{4}(-\underline{\sigma}'  - \omega \lambda ),
\end{equation}
where  $\omega= \int_{\mathbb{R}^{+}} 
\left(e^{\rho x}-1 \right) \nu_Z (dx)$.
\end{theorem}
\begin{proof}
Recall that $$\hat{Y}_t = S_0 + \sigma S_0 {W}_t + S_0 
\int_{0}^{t}\int_{\mathbb{R}^{+}} \left(e^{\rho x}-1 \right) 
\tilde{J}_Z(\lambda du, dx). $$
Hence, 
\begin{align*}
\label{thm4eq60}
\frac{1}{T}\int_{0}^{T} \hat{Y}_t dt - S_0 
& = \frac{1}{T}\int_{0}^{T} \left( S_0 + \sigma S_0 {W}_t + S_0 
\int_{0}^{t}\int_{\mathbb{R}^{+}} \left(e^{\rho x}-1 \right) 
\tilde{J}_Z(\lambda du, dx)\right) dt - S_0 \\ 
& = \sigma S_0 \frac{1}{T}\int_{0}^{T} {W}_t dt +  S_0 
\frac{1}{T}\int_{0}^{T} \int_{0}^{t}\int_{\mathbb{R}^{+}} \left(e^{\rho x}-1\right) \tilde{J}_Z(\lambda du, dx)dt.
\end{align*}
Therefore,
\begin{align*}
\mathbb{E} \left[ \left( \frac{1}{T}\int_{0}^{T} \hat{Y}_t dt - S_0 \right)^{+} 
\right] = 
\frac{1}{2} \mathbb{E} \left[\sigma S_0 \frac{1}{T}\int_{0}^{T} {W}_t dt +  S_0 
\frac{1}{T}\int_{0}^{T} \int_{0}^{t}\int_{\mathbb{R}^{+}} \left(e^{\rho x}-1 \right) \tilde{J}_Z(\lambda du, dx)dt \right] + \\
\frac{1}{2} \mathbb{E} \left[ \left|\sigma S_0 \frac{1}{T}\int_{0}^{T} {W}_t dt 
+  S_0 \frac{1}{T}\int_{0}^{T} \int_{0}^{t}\int_{\mathbb{R}^{+}} 
\left(e^{\rho x}-1 \right) \tilde{J}_Z(\lambda du, dx)dt 
\right| \right].
\end{align*}
Observe that the first term on the right side of the above equality zero since 
the Brownian motion $W$ and the compensated random jump measure $\tilde{J}_Z$ 
are martingales, which leaves us 
with
\begin{equation}
	\label{thm5eqn58}
\begin{split}
\mathbb{E} \left[ \left( \frac{1}{T}\int_{0}^{T} \hat{Y}_t dt - S_0 \right)^{+} 
\right] = \frac{1}{2} \mathbb{E} \left[ \left|\sigma S_0 
\frac{1}{T}\int_{0}^{T} {W}_t dt +  S_0 \frac{1}{T}\int_{0}^{T} 
\int_{0}^{t}\int_{\mathbb{R}^{+}} \left(e^{\rho x}-1 
\right) \tilde{J}_Z(\lambda du, dx)dt \right| \right]. 
\end{split}
\end{equation}
It follows from \eqref{thm5eqn58} that:
\begin{align}
\label{east}
\mathbb{E} \left[ \left( \frac{1}{T}\int_{0}^{T} \hat{Y}_t dt - S_0 \right)^{+} 
\right] & \leq \frac{1}{2} \frac{\sigma S_0 }{T}\mathbb{E} \left|
\int_{0}^{T} {W}_t dt\right|  + \frac{1}{2} \frac{S_0}{T}\mathbb{E}  \int_{0}^{T}\left| 
\int_{0}^{t}\int_{\mathbb{R}^{+}} \left(e^{\rho x}-1 
\right) \tilde{J}_Z(\lambda du, dx)dt \right|.
\end{align}
As $\int_{0}^{T} {W}_t dt \sim N(0, \frac{T^3}{3})$, $\mathbb{E} \left|
\int_{0}^{T} {W}_t dt\right|= \sqrt{\frac{T^3}{3}} \sqrt{\frac{2}{\pi}}$.
Hence, using this observation, and \eqref{kbkbkb}, we obtain the following from \eqref{east}:
\begin{align}
\label{east1}
\mathbb{E} \left[ \left( \frac{1}{T}\int_{0}^{T} \hat{Y}_t dt - S_0 \right)^{+} 
\right] \leq  \frac{\sigma S_0 }{\sqrt{6 \pi}}\sqrt{T} + \frac{S_0 T}{4}(-\underline{\sigma}'  - \omega \lambda ).
\end{align}
On the other hand,  it also follows from \eqref{thm5eqn58} and \eqref{kbkbkb} that:
\begin{align}
\label{west}
\mathbb{E} \left[ \left( \frac{1}{T}\int_{0}^{T} \hat{Y}_t dt - S_0 \right)^{+} 
\right] & \geq \frac{1}{2} \frac{\sigma S_0 }{T}\mathbb{E} \left|
\int_{0}^{T} {W}_t dt\right|  - \frac{1}{2} \frac{S_0}{T}\mathbb{E} \left| \int_{0}^{T} 
\int_{0}^{t}\int_{\mathbb{R}^{+}} \left(e^{\rho x}-1 
\right) \tilde{J}_Z(\lambda du, dx)dt \right| \nonumber \\
& \geq  \frac{\sigma S_0 }{\sqrt{6 \pi}}\sqrt{T}  - \frac{1}{2} \frac{S_0}{T}\mathbb{E} \int_{0}^{T} \left| 
\int_{0}^{t}\int_{\mathbb{R}^{+}} \left(e^{\rho x}-1 
\right) \tilde{J}_Z(\lambda du, dx)dt \right| \nonumber \\
& \geq \frac{\sigma S_0 }{\sqrt{6 \pi}}\sqrt{T} - \frac{S_0 T}{4}(-\underline{\sigma}'  - \omega \lambda ).
\end{align}
By \eqref{east1} and \eqref{west} we obtain \eqref{ayo}.
\end{proof}
Now we state and prove the main theorem. 
\begin{theorem}
\label{thm34}
Suppose that Assumptions \ref{as512}, \ref{ch5as5224}, and \ref{asmp528} hold 
true. 
Then,
\begin{enumerate}
\item[(i)]  when $K = S_0$, as $T \to 0,$
\begin{equation}
\label{thm34new2}
\left(\frac{\sigma S_0 }{\sqrt{6 \pi}}\sqrt{T} - \frac{S_0 T}{4}(-\underline{\sigma}'  - \omega \lambda )\right)^+   + O(T) \leq C(T) \leq   \frac{\sigma S_0 }{\sqrt{6 \pi}}\sqrt{T} + \frac{S_0 T}{4}(-\underline{\sigma}'  - \omega \lambda )+ O(T).
\end{equation}
\item[(ii)] When $K = S_0$, as $T \to 0,$
\begin{equation}
\label{thm34new3}
\left(\frac{\sigma S_0 }{\sqrt{6 \pi}}\sqrt{T} - \frac{S_0 T}{4}(-\underline{\sigma}'  - \omega \lambda )\right)^+  + O(T) \leq P(T) \leq  \frac{\sigma S_0 }{\sqrt{6 \pi}}\sqrt{T} + \frac{S_0 T}{4}(-\underline{\sigma}'  - \omega \lambda ) + O(T).
\end{equation}
\end{enumerate}
where $\omega = 	
\int_{\mathbb{R}^{+}} \left(e^{\rho x}-1\right) \nu_{Z} (dx)$.
\end{theorem}

\begin{proof}
For ATM call option, recall that
\begin{align*}
C(T)= e^{-rT}\mathbb{E}  \left[ \left( \frac{1}{T}\int_{0}^{T}	S_t 
dt - S_0 \right)^{+} \right].
\end{align*}	
By Assumption \ref{asmp528}, with respect to risk-neutral measure $\mathbb{P}$,  $Y_t = e^{-(r-q)t} S_t$ satisfies the SDE 
\begin{equation*}
dY_t = \sigma Y_t d W_t + Y_t \int_{\mathbb{R}^{+}} \left(e^{\rho x}-1\right) \tilde{J}_Z(\lambda dt, dx), \quad \mathrm{with} \quad Y_0 = S_0, 
\end{equation*}
where $W_t$ denotes the Brownian motion, $\tilde{J}_Z$ is the compensated 
random jump measure of the L\'evy subordinator $Z$ where $W, Z$ are assumed to 
be independent. Therefore, for ATM call option, 
\begin{equation*}
\label{thm34eqn1}
C(T)= e^{-rT}\mathbb{E}  \left[ \left(\frac{1}{T}\int_{0}^{T} 
e^{(r-q)t} Y_t dt - S_0 \right)^{+} \right].
\end{equation*}
Putting together the results of Lemma \ref{Claim 1}, Theorem \ref{Claim 2}, Lemma \ref{Claim 3}, and Theorem \ref{Claim 4}, we obtain the result for the right side of the inequality of \eqref{thm34new2}.  This proves our desired 
result for ATM call Asian options. The result \eqref{thm34new3} follows similarly by Put-Call parity. 
\end{proof}

\begin{remark}
It is worth noting that for the ``no-jump" case, Theorem \ref{thm34} gives $ C(T) =   \frac{\sigma S_0 }{\sqrt{6 \pi}}\sqrt{T} + O(T)$, and
$ P(T) =  \frac{\sigma S_0 }{\sqrt{6 \pi}}\sqrt{T} + O(T)$. These results agree with \cite{pirjol} for the case of constant volatility. 

\end{remark}

\section{Numerical results and conclusion}
\label{sec5}
In this section we provide some numerical results based on the present analysis. We consider three stocks: Apple (``AAPL"), Walmart (``WMT"), and Tesla (``TSLA"), starting on January 1, 2023. The data is obtained from \url{https://finance.yahoo.com/}. For each stock we use daily \emph{open} price of the stocks. For the computation we use $TA(T)$, where $A(T)$ can be obtained from \eqref{lemma52111}, for different $C$ values. $T$ specifies the number of days. In the following examples $T= 30$, $T=14$, and $T=7$ are used. The parameters $r$, $q$, and $\sigma^2$ are obtained from the respective datasets. We take $\rho= -1$ and assume that the L\'evy measure of the subordinator $Z$ is given by $\nu_Z(dx)= \frac{\delta}{2\sqrt{2 \pi}} x^{-\frac{3}{2}} (1+ \gamma^2 x)e^{-\frac{1}{2} \gamma^2 x}dx$, for $x \geq 0$, and $\gamma, \delta \geq 0$. This is obtained from some modification of $\text{IG}(\delta, \gamma)$ distribution (see \cite{habt, ijtaf}). Clearly, all the assumptions made in this paper are satisfied by this subordinator. The fitting parameter $C$ is obtained with a grid-search on $\delta$ and $\gamma$. With these appropriate parameters (especially the value of $C$, that corresponds to the jump), the model \eqref{ch5sec21} is better understood, and the results related to the Asian options (described in the previous sections) are derivable in the numerical sense. 

For the following plots we take $r-q<0$.  The following results will show a justification for the need to incorporate the jump terms in the analysis of \cite{pirjol}.  For a 30-day time period, Figure \ref{pic3} provides a graphical representation of the integrated stock value (i.e., $\int_0^T S_t\,dt$), and its deviation from the model with no jump ($C=0$) and a specific jump ($C=0.1$) case. The best $C$-value fit is shown in Figure \ref{pic4}. Figure \ref{pic4} also shows the corresponding RMSE. 
\begin{figure}[H]
\centering
\includegraphics[width=1\textwidth]{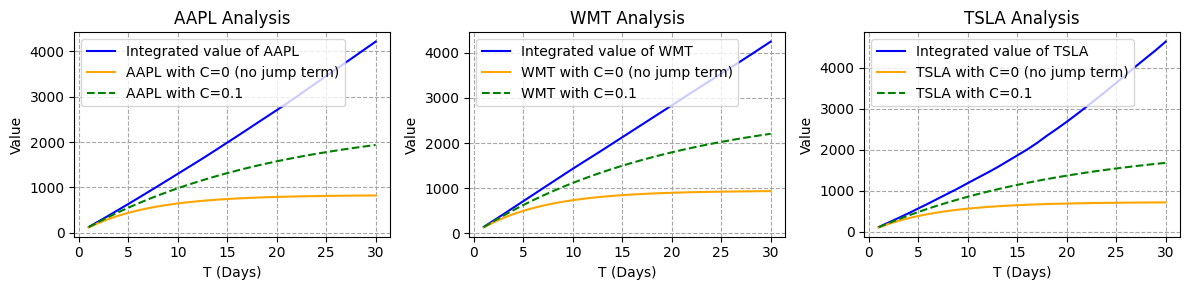}
\caption{Integrated stock values, with values from model with zero and non-zero $C$.}
\label{pic3}
\end{figure}
\begin{figure}[H]
\centering
\includegraphics[width=.9\textwidth]{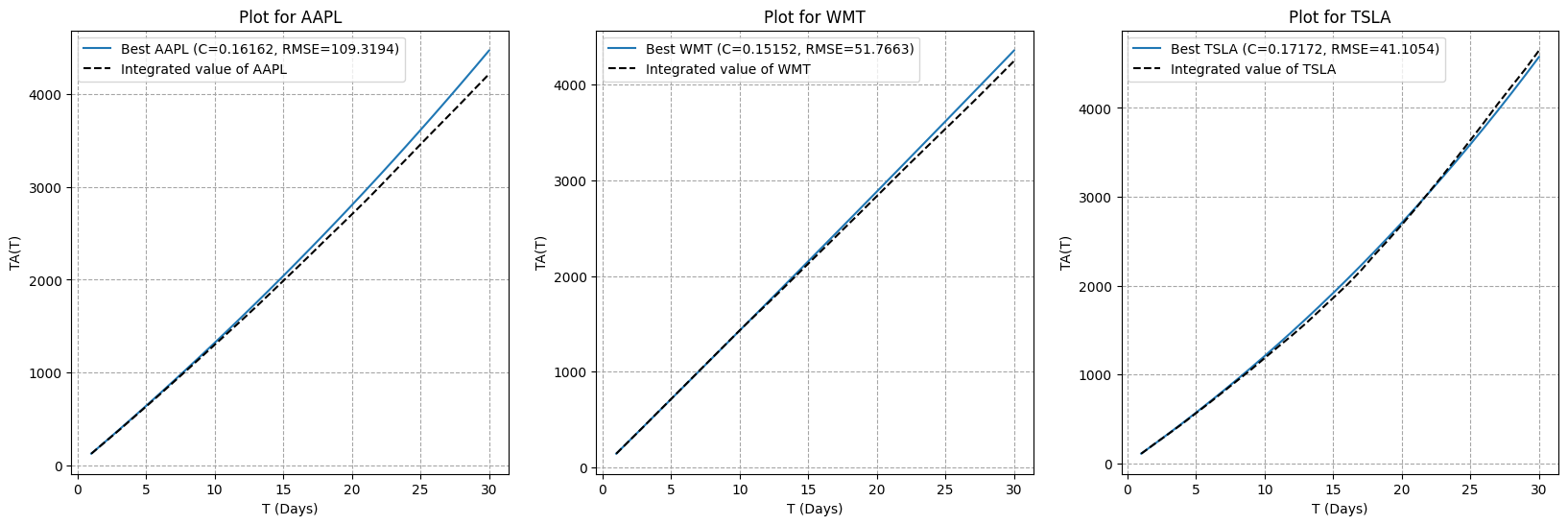}
\caption{Fitting of integrated stock values with best $C$.}
\label{pic4}
\end{figure}
For a 14-day time period, Figure \ref{pic1} provides a graphical representation of the integrated stock value and its deviation from the model with no jump ($C=0$) and a specific jump ($C=0.1$) case. The best $C$-value fit is shown in Figure \ref{pic2}. Figure \ref{pic2} also shows the corresponding RMSE. 
\begin{figure}[H]
\centering
\includegraphics[width=1\textwidth]{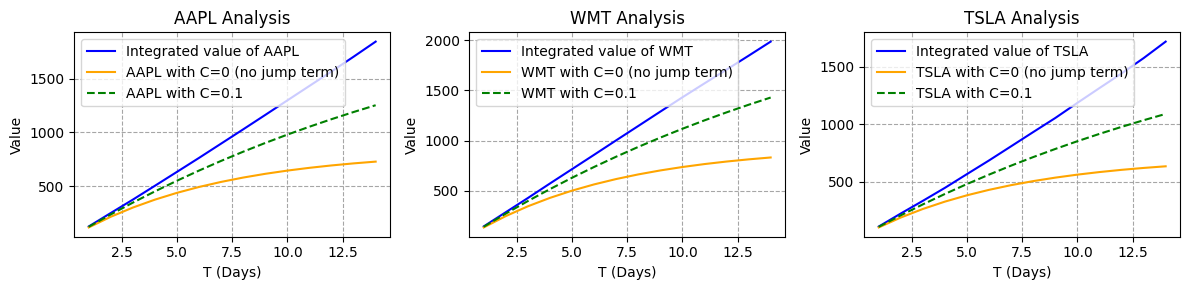}
\caption{Integrated stock values, with values from model with zero and non-zero $C$.}
\label{pic1}
\end{figure}
\begin{figure}[H]
\centering
\includegraphics[width=.9\textwidth]{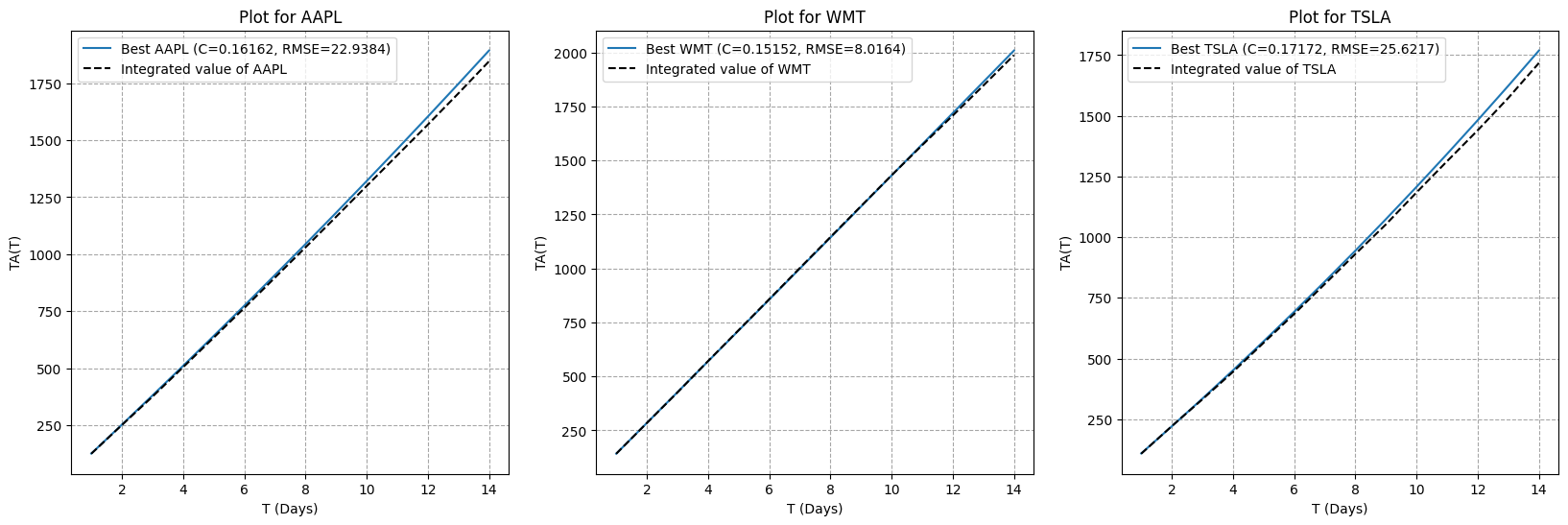}
\caption{Fitting of integrated stock values with best $C$.}
\label{pic2}
\end{figure}
For a 7-day time period, Figure \ref{pic5} provides a graphical representation of the integrated stock value and its deviation from the model with no jump ($C=0$) and a specific jump ($C=0.1$) case. The best $C$-value fit is shown in Figure \ref{pic6}. Figure \ref{pic6} also shows the corresponding RMSE. 
\begin{figure}[H]
\centering
\includegraphics[width=1\textwidth]{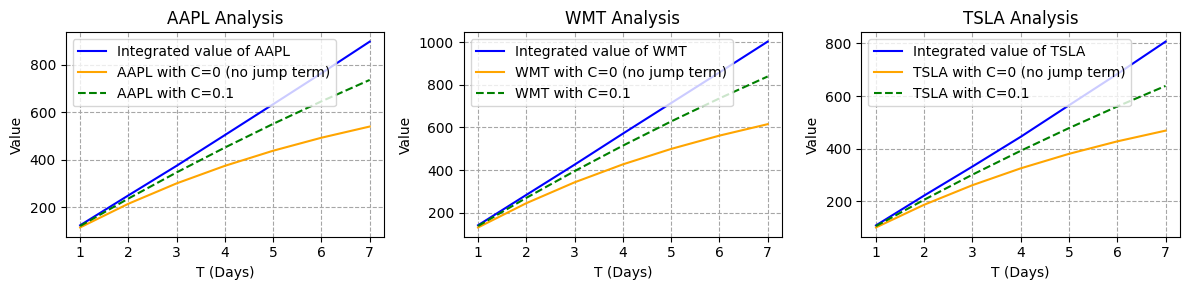}
\caption{Integrated stock values, with values from model with zero and non-zero $C$.}
\label{pic5}
\end{figure}
\begin{figure}[H]
\centering
\includegraphics[width=.9\textwidth]{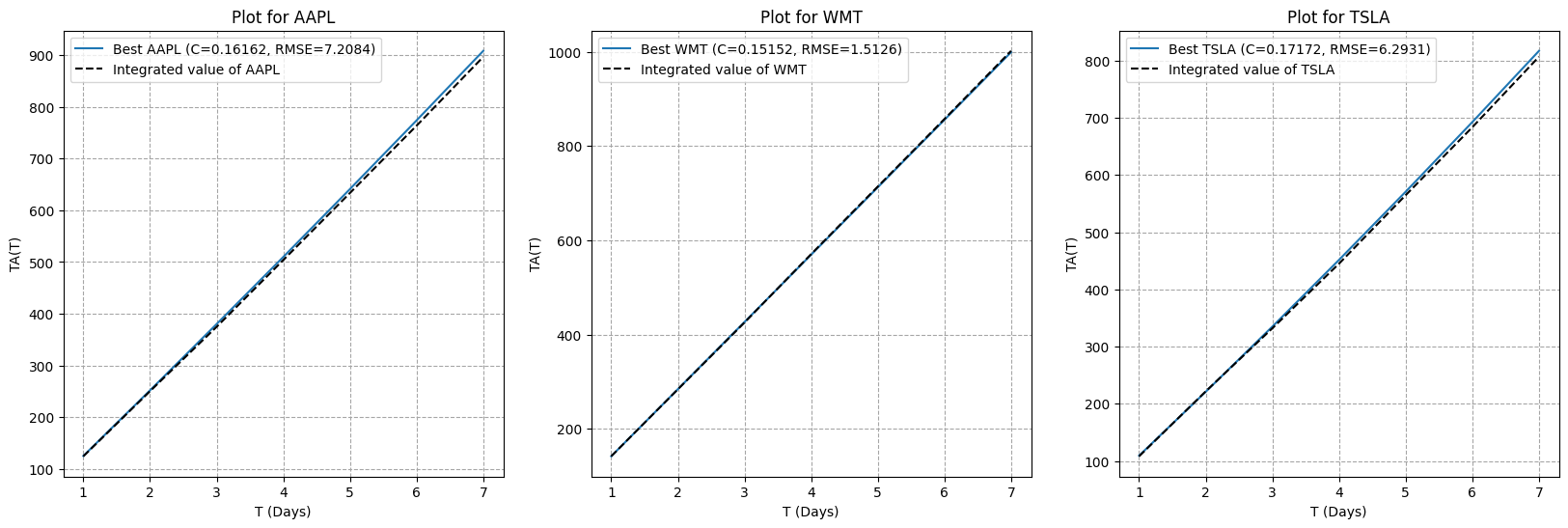}
\caption{Fitting of integrated stock values with best $C$.}
\label{pic6}
\end{figure}
It is clear from the results that the results in the non-zero $C$ cases (i.e., models with jump) are better than the $C=0$ (i.e., no jump) case. The results support the rationale for integrating jumps into the analysis outlined in this manuscript.  

In summary, motivated from previous study done with short maturity Asian options in local volatility models, where the dynamics of the underlying asset price was modeled using a drift term and a Brownian motion term, in this paper, we have presented mathematical works for a constant volatility model that incorporates a jump term in addition to the drift and diffusion terms. We 
estimated the asymptotics for the out-of-the-money, in-the-money, and at-the-money short maturity Asian call and put options where we showed that the asymptotics for out-of-the-money Asian call and put options are governed by 
rare events (large deviation). From these results, the asymptotics for short maturity in-the-money Asian call and put options can 
be obtained by the application of the put-call parity. For the case of at-the-money Asian options the upper and lower bound for asymptotics of option prices are found.  

\vspace{.1in}
\textbf{Data Availability Statement}:  The data that support the findings of this study are openly available in Yahoo Finance at \url{https://finance.yahoo.com}.

\vspace{.1in}
\textbf{Declaration} (for H. Shoshi): 
The views expressed herein are for informational purposes only and are those of the author and do not necessarily reflect the views of Citigroup Inc. All opinions are subject to change without notice.

\vspace{.1in}
\textbf{Acknowledgments}: We thank the associate editor and two anonymous referees for their constructive comments and suggestions, which helped us improve the manuscript.


\begin{thebibliography}{99} \footnotesize


\bibitem{alos} 
Al\'os E., L\'eon J., and Vives J., (2007), On the short-time behavior of the implied volatility for jump-diffusion models with stochastic volatility, \emph{Finance Stoch.}, \textbf{11}, 571-589.


\bibitem{aplbm}
Applebaum D., (2009), \emph{L\'evy Processes and Stochastic Calculus}, 2nd ed., 
Cambridge University Press, Cambridge, UK.


\bibitem{bbf}
Berestycki H., Busca J., and Florent I., (2002), Asymptotics and calibration of local volatility models, \emph{Quant. Finance}, \textbf{2}, 61-69.




\bibitem{carr}
Carr P., and Schr\"oder M., (2003), Bessel processes, the integral of geometric Brownian motion, and Asian options, \emph{Theory Probab. Appl.}, \textbf{48}, 400-425.


\bibitem{sudip} Chandra S. R., Mukherjee D., and SenGupta I. (2015), PIDE and Solution Related to Pricing of Lévy Driven Arithmetic Type Floating Asian Options, \emph{Stochastic Analysis and Applications}, \textbf{33} (4), 630-652. 

\bibitem{chengetal}
Cheng W., Costanzino N., Liechty J., Mazzucato A., and Nistor V., (2011), Closed-form asymptotics and numerical approximations of 1 D parabolic equations with applications to option pricing, \emph{SIAM J. Financial Math.}, \textbf{2}, 901-934.



\bibitem{cont}
Cont R., and Tankov P., (2004), \emph{Financial Modelling with Jump Processes}, CRC Financial Mathematics Series, Boca Raton, Chapman and~Hall.

\bibitem{dembo}
Dembo A., and Zeitouni O., (1998), \emph{Large Deviations Techniques and Applications}, 2nd ed., (eds Berlin/ Heidelberg), Springer.


\bibitem{dufrense}
Dufrense D., (2000), Laguerre series for Asian and other options, \emph{Math. Finance}, \textbf{10}, 407-428.

\bibitem{figlopez}
Figueroa-L\'opez J. E., and Forde M., (2012), The small-maturity smile for exponential L\'evy models, \emph{SIAM J. Financial Math.}, \textbf{3}, 33-65.

\bibitem{mforde}
Forde M., and Jacquier A., and Lee R., (2012), The small-time smile and term structure of implied volatility under the Heston model, \emph{SIAM J. Financial Math.}, \textbf{3}, 690-708.

\bibitem{foschip}
Foschi P., Pagliarani S., and Pascucci A., (2013), Approximations for Asian options in local volatility models, \emph{J. Comput. Appl. Math.}, \textbf{237}, 442-459.

\bibitem{fuetal}
Fu M., Madan D., and Wang T., (1998), Pricing continuous time Asian options: A 
comparison of Monte Carlo and Laplace transform inversion methods, \emph{J. 
Comput. Finance}, \textbf{2}, 49-74.

\bibitem{geman}
Geman H., and Yor M., (1993), Bessel processes, Asian options and perpetuities, \emph{Math. Finance}, \textbf{3}, 349-375.



\bibitem{habt}
Habtemicael S., and SenGupta I., (2016), Pricing variance and volatility swaps for Barndorff-Nielsen and Shephard process driven financial markets,
\emph{International Journal of Financial Engineering}, \textbf{03}(04), 1650027 (35 pages).

\bibitem{kuhn} 
K\"uhn F., (2014), \textit{Large Deviations for L\'evy(-Type) Processes}, PhD 
Thesis, Institut fur Mathematische Stochastik.


\bibitem{linetsky}
Linetsky V., (2004), Spectral expansions for Asian (Average price) options, \emph{Oper. Res.}, \textbf{52}, 856-867.


\bibitem{NV} Nicolato E., and Venardos E., (2003), Option Pricing in Stochastic Volatility Models of the Ornstein-Uhlenbeck type, \emph{Mathematical Finance}, \textbf{13} (4), 445-466.


\bibitem{pirjol}
Pirjol D., and Zhu L., (2016), Short Maturity Asian Options in Local Volatility Models, \emph{SIAM J. Financial Math.}, \textbf{7}(1).

\bibitem{pirjol2}
Pirjol D., and Zhu L., (2017), Asymptotics for the Discrete-Time Average of the Geometric Brownian Motion and Asian Options, \emph{Advances in Applied Probability}, \textbf{49}, 2, 446-480.





\bibitem{rogersshi}
Rogers L., and Shi Z.,(1995), The value of an Asian option, \emph{J. Appl. 
Probab.}, \textbf{32}, 1077-1088.

\bibitem{Mellin} SenGupta I. (2014), Pricing Asian options in financial markets using Mellin transforms, \emph{Electronic Journal of Differential Equations}, \textbf{2014} (234), 1-9.

\bibitem{ijtaf} 
SenGupta I.,(2016), Generalized BN-S stochastic volatility model for option 
pricing, \emph{International Journal of Theoretical and Applied Finance}, 
\textbf{19}(02), 1650014 (23 pages).





\bibitem {hum}
Shoshi  H., and  SenGupta I., (2021), Hedging and Machine Learning Driven Crude 
Oil Data Analysis Using a Refined Barndorff--Nielsen and Shephard Model, 
\emph{International Journal of Financial Engineering}, \textbf{8}, 2150015. 

\bibitem{hum2}
Shoshi H., Hanson E., Nganje W., and SenGupta I., (2021), Stochastic analysis 
and neural network based yield prediction with precision agriculture, 
\emph{Journal of Risk and Financial Management}, \textbf{14} (9), 397




\bibitem{vara}
Varadhan S.R.S., (1967), Diffusion Processes in a Small Time Interval, 
\emph{Communications on pure and applied mathematics}, \textbf{XX}, 659-685.

\bibitem{vecer1}
Vecer J., (2001), A new PDE approach for pricing arithmetic average Asian options, \emph{J. Comput. Finance}, \textbf{4}, 105-113.

\bibitem{vecer2}
Vecer J., and Xu M., (2002), Unified Asian pricing, \emph{Risk}, \textbf{15}, 113-116.








                   

\end{thebibliography}
\end{document}